\documentclass[11pt]{article}
\usepackage[margin=1in]{geometry}

\usepackage{graphicx}
\usepackage{amssymb,amsmath,amsfonts}
\usepackage{amsthm}
\usepackage{algorithm}
\usepackage{algpseudocode}
\usepackage{algpascal}
\usepackage{bm}
\usepackage{siunitx}
\usepackage{hyperref}
\usepackage[caption=false,font=footnotesize]{subfig}
\usepackage{xcolor}
\usepackage{tikz}
\usetikzlibrary{arrows,patterns}
\usetikzlibrary{positioning}
\usetikzlibrary{backgrounds}
\usetikzlibrary{arrows.meta}
\usepackage{cite}

\newcommand{\minimize}{\mathop{\text{minimize}}}

\newcommand{\R}{\mathbb{R}}
\newcommand{\eg}{\textit{e.g.}}
\newcommand{\ie}{\textit{i.e.}}
\newcommand{\defeq}{:=}

\newcommand{\thesetof}[2]{\left\{ #1 \;\middle\vert\; #2 \right\}}
\newcommand{\norm}[1]{\left\Vert #1 \right\Vert}
\newcommand{\abs}[1]{\left\vert #1 \right\vert}
\newcommand{\T}{\top}
\newcommand{\argmin}{\mathop{\operatorname{argmin}}}
\newcommand{\zeros}{\bm{0}}
\newcommand{\ones}{\bm{1}}

\newcommand{\xstar}{x^{\star}}
\newcommand{\Nin}{N^{\text{in}}}
\newcommand{\Nout}{N^{\text{out}}}

\newcommand{\xbar}{\overline{x}}

\newcommand{\Pbar}{\overline{\bm{P}}}

\newcommand{\Ptilde}[1]{\widetilde{\bm{P}}_{_{ #1 }}}

\newcommand{\Graph}{\mathcal{G}}
\newcommand{\Edges}{\mathcal{E}}
\newcommand{\edge}[2]{(v_{#2} \leftarrow v_{#1})}
\newcommand{\vedge}[4]{(v^{(#4)}_{#3} \leftarrow v^{(#2)}_{#1})}
\newcommand{\Vertices}{\mathcal{V}}
\newcommand{\tautranmax}{\overline{\tau}^{\text{msg}}}
\newcommand{\tauratemax}{\overline{\tau}^{\text{proc}}}
\newcommand{\taumax}{\overline{\tau}}
\newcommand{\tautran}[3]{\tau^{\text{msg}}_{ #2 #3 }[ #1 ]}
\newcommand{\taurate}[2]{\tau^{\text{proc}}_{ #2 }[ #1 ]}
\newcommand{\Gset}[1]{\mathcal{A}[ #1 ]}

\newcommand{\Grad}[1]{\nabla \bm{\overline{F}}[ #1 ]}
\newcommand{\Eye}[1]{\bm{I}_{_{#1 \times #1}}}

\newtheorem{assumption}{Assumption}
\newtheorem*{assumption*}{Assumption}
\newtheorem{theorem}{Theorem}
\newtheorem{definition}{Definition}

\newtheorem{lemma}{Lemma}
\newtheorem{corollary}{Corollary}[theorem]
\newtheorem{remark}{Remark}
\newtheorem*{remark*}{Remark}
\newtheorem*{remarkth*}{Remark}

\begin{document}	
\title{Asynchronous Gradient-Push}
\author{Mahmoud Assran and Michael Rabbat%
\thanks{The authors are with Facebook AI Research, Montr\'{e}al, Qu\'{e}bec, Canada, and the Department of Electrical and Computer Engineering, McGill University, Montr\'{e}al, Qu\'{e}bec, Canada. Email: mahmoud.assran@mail.mcgill.ca, mikerabbat@fb.com.
}}
\maketitle
 
\begin{abstract}
We consider a multi-agent framework for distributed optimization where each agent has access to a local smooth strongly convex function, and the collective goal is to achieve consensus on the parameters that minimize the sum of the agents' local functions.  We propose an algorithm wherein each agent operates asynchronously and independently of the other agents. When the local functions are strongly-convex with Lipschitz-continuous gradients, we show that the iterates at each agent converge to a neighborhood of the global minimum, where the neighborhood size depends on the degree of asynchrony in the multi-agent network. When the agents work at the same rate, convergence to the global minimizer is achieved. Numerical experiments demonstrate that Asynchronous Gradient-Push can minimize the global objective faster than state-of-the-art synchronous first-order methods, is more robust to failing or stalling agents, and scales better with the network size.
\end{abstract}

\section{Introduction}

We propose and analyze an asynchronous distributed algorithm to solve the optimization problem
\begin{equation} \label{eq:problem}
\begin{array}{l l}
	\minimize_{x \in \R^d} & F(x) \defeq \sum_{i=1}^n f_i(x)
\end{array}
\end{equation}
where each $f_i:\R^d \to \R$ is smooth and strongly convex. We focus on the multi-agent setting, in which there are $n$ agents and information about the function $f_i$ is only available at the $i^{\text{th}}$ agent. Specifically, only the $i^{\text{th}}$ agent can evaluate $f_i$ and gradients of $f_i$. Consequently, the agents must cooperate to find a minimizer of $F$.

Many multi-agent optimization algorithms have been proposed, motivated by a variety of applications including distributed sensing systems, the internet of things, the smart grid, multi-robot systems, and large-scale machine learning. In general, there have been significant advances in the development of distributed methods with theoretical convergence guarantees in a variety of challenging scenarios such as time-varying and directed graphs (see~\cite{nedic2017network} for a recent survey). However, the vast majority of this literature has focused on \emph{synchronous} methods, where all agents perform updates at the same rate.

This paper studies asynchronous distributed algorithms for multi-agent optimization.
Our interest in this setting comes from applications of multi-agent methods to solve large-scale optimization problems arising in the context of machine learning, where each agent may be running on a different server and the agents communicate over a wired network. Hence, agents may receive multiple messages from their neighbours at any given time instant, and may perform a drastically different number of gradient steps over any time interval. In distributed computing systems, communication delays may be unpredictable; communication links may be unreliable; 
and each processor may be shared for other tasks while at the same time cooperating with other processors in the context of some computational task~\cite{bertsekas1989parallel}. High performance computing clusters fit this model of distributed computing quite nicely~\cite{tsianos2012communication}, especially since node and link failures may be expected in such systems \cite{tsianos2012consensus, kempe2003gossip, dean2013tail}. When a synchronous algorithm is run in such a setting, the rate of progress of the entire system is hampered by the slowest node or communication link; asynchronous algorithms are largely immune to such issues~\cite{rabbat2014asynchronous, bertsekas1989parallel, Lian2018asynchronous, assran2017empirical, cannelli2017asynchronous, hale2017asynchronous, kumar2017asynchronous, aytekin2017asynchronous, wu2016decentralized}.

\subsection{Asynchronous Gradient-Push}

Practical implementations of multi-agent communication---using the Message Passing Interface (MPI)~\cite{gropp1996high} or other message passing standards---often have the notion of a \emph{send-buffer} and a \emph{receive-buffer}.  A send-buffer is a data structure containing the messages sent by an agent, but not yet physically transmitted by the underlying communication system.  A receive-buffer is a data structure containing the messages received by an agent, but not yet processed by the application. 

Using this notion of send- and receive-buffers, the individual-agent pseudocode for running the asynchronous gradient-push method is shown in Algorithm~\ref{alg:agp}. The method repeats a two-step procedure consisting of \textbf{Local Computation} followed by \textbf{Asynchronous Gossip}. During the \textbf{Local Computation} phase, agents update their estimate of the minimizer by performing a local (sub)gradient-descent step. During the \textbf{Asynchronous Gossip} phase, agents copy all outgoing messages into their local send-buffer and subsequently process (sum) all messages received (buffered) in their local receive-buffer while the agent was busy performing the preceding \textbf{Local Computation}. The underlying communication system begins transmitting the messages in the local send-buffer once they are copied there; thereby freeing the agent to proceed to the next step of the algorithm without waiting for the messages to reach their destination.

\begin{figure}[!t]
\alglanguage{pseudocode}
\begin{algorithm}[H]
\caption{Asynchronous Gradient-Push (Pseudocode) for agent $v_i$} \label{alg:agp}
\begin{algorithmic}[1]
\State Initialize $x_i \in \R^d$ \Comment{Push-sum numerator}
\State Initialize $y_i \gets 1$ \Comment{Push-sum weight}
\State Initialize $\alpha_i > 0$ \Comment{Step-size}
\State $\Nout_i \gets \text{number of out-neighbours of } v_i$
\While{stopping criterion not satisfied}{}
	\State \textbf{Begin: Local Computation}
	    \State $z_i \gets x_i / y_i$ \Comment{De-biased consensus estimate}
	    \State $x_i \gets x_i - \alpha_i \nabla f_i(z_i)$ \label{lst:gradient_update}
		\State Update step-size $\alpha_i$
	\State \textbf{Begin: Asynchronous Gossip}
    	\State Copy message $(x_i / \Nout_i, y_i / \Nout_i)$ to local send-buffer
	\State $x_i \gets x_i/\Nout_i + \sum_{(x', y') \in \text{receive-buffer}} x'$
	\State $y_i \gets y_i/\Nout_i + \sum_{(x', y') \in \text{receive-buffer}} y'$
\EndWhile
\end{algorithmic}
\end{algorithm}
\end{figure}

Fig.~\ref{fig:ssp_diag} illustrates the agent update procedure in the synchronous case: agents must wait for all network communications to be completed before moving-on to the next iteration, and, as a result, some agents may experience idling periods. Fig.~\ref{fig:asp_diag} illustrates the agent update procedure in the asynchronous case: at the beginning of each local iteration, agents make use of their message buffers by copying all outgoing messages into their local send-buffers, and by retrieving all messages from their local receive-buffers. The underlying communication systems subsequently transmit the messages in the send-buffers while the agents proceed with their computations.

\begin{figure}[!t]
\centering
\subfloat[\label{fig:ssp_diag} ]{\includegraphics[width=\columnwidth]{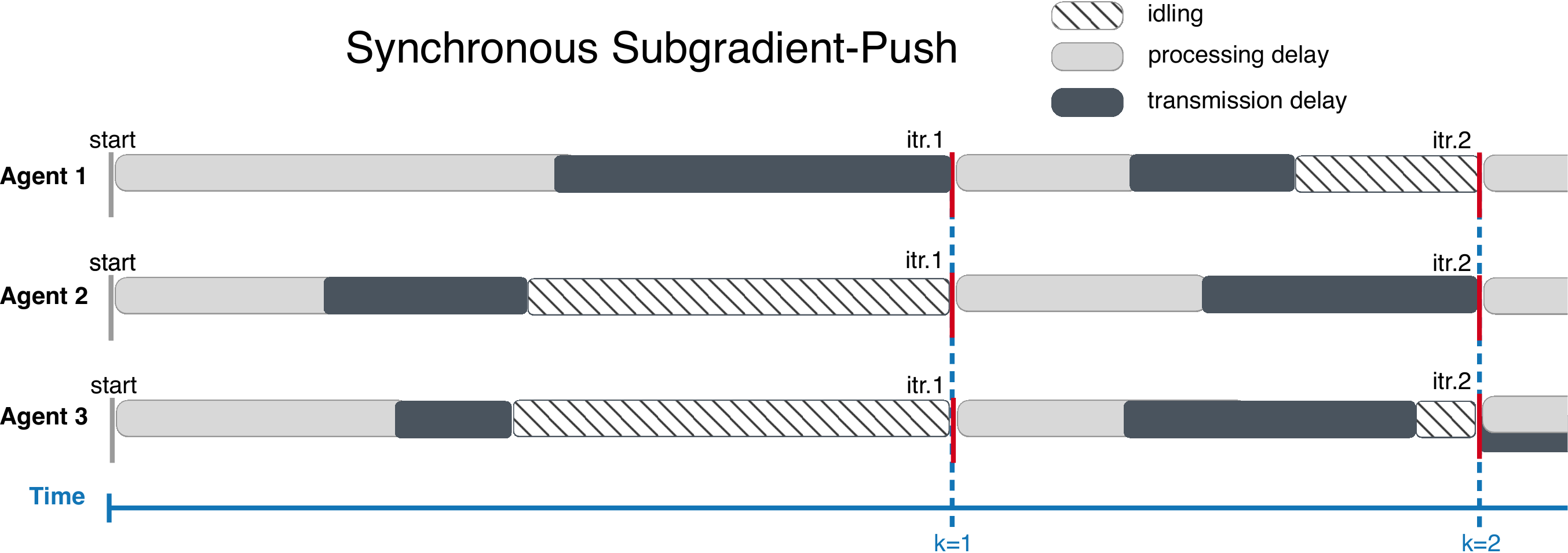}} \\
\subfloat[\label{fig:asp_diag}]{\includegraphics[width=\columnwidth]{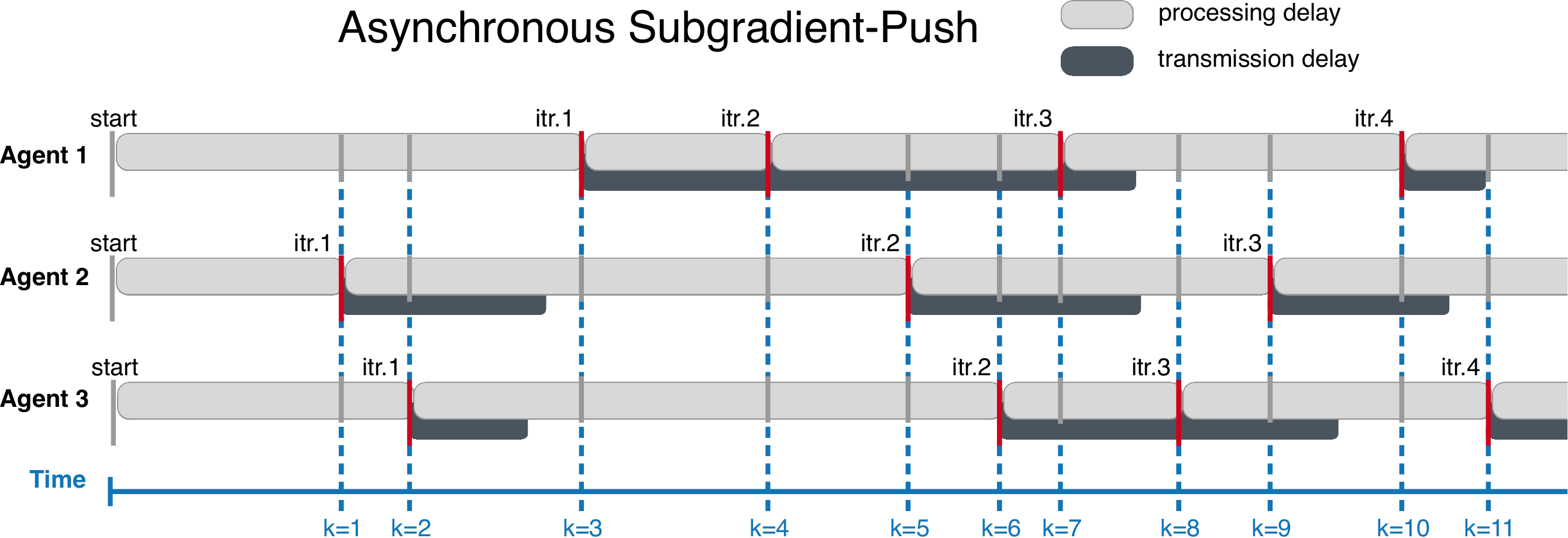}}
\caption{Example of agent updates in synchronous and asynchronous Gradient-Push implementations. Processing delays correspond to the time required to perform a local iteration. Transmission delays correspond to the time required for all outgoing message to arrive at their destination buffers. Even though a message arrives at a destination agent's receive-buffer after some real (non-integer valued) delay, that message is only processed when the destination agents performs its next update.}
\end{figure}

\subsection{Related Work}
\label{sec:related-work}

Most multi-agent optimization methods are built on distributed averaging algorithms~\cite{nedic2007rate}. For synchronous methods operating over static, undirected networks, it is possible to use doubly stochastic averaging matrices. However, it turns out that averaging protocols which rely on doubly stochastic matrices may be undesirable for a variety of reasons~\cite{tsianos2012consensus}. The Push-Sum approach to distributed averaging, introduced in~\cite{kempe2003gossip}, eliminates the need for doubly stochastic consensus matrices. The seminal work on Push-Sum~\cite{kempe2003gossip} analyzed convergence for complete network topologies (all pairs of agents may communicate directly). The analysis was extended in~\cite{benezit2010weighted} for general connected graphs. Further work has provided convergence guarantees in the face of the other practical issues, such as communication delays and dropped messages~\cite{charalambous2015distributed, hadjicostis2014average}. In general, Push-Sum is attractive for implementations because it can easily handle directed communication topologies, and thus avoids incidents of deadlock that may occur in practice when using undirected communication topologies~\cite{tsianos2012consensus}.

\subsubsection*{Multi-Agent Optimization with Column Stochastic Consensus Matrices}
The first multi-agent optimization algorithm using Push-Sum for distributed averaging was proposed in~\cite{tsianos2012push}. Nedi\'{c} and Olshevsky~\cite{nedich2015distributed} continue this line of work by introducing and analyzing the Subgradient-Push method; the analysis in~\cite{nedich2015distributed} focuses on minimizing (weakly) convex, Lipschitz functions, for which diminishing step-sizes are required to obtain convergence. Xi and Khan~\cite{xi2017dextra} propose DEXTRA and Zeng and Yin~\cite{zeng2015extrapush} propose Extra-Push, both of which use the Push-Sum protocol in conjunction with gradient tracking to achieve geometric convergence for smooth, strongly convex objectives over directed graphs. Nedi\'{c}, Olshevsky, and Shi~\cite{nedic2017achieving} propose the Push-DIGing algorithm, which achieves a geometric convergence rate over directed and time-varying communication graphs. Push-DIGing and DEXTRA/Extra-Push are considered to be state-of-the-art synchronous methods, and the Subgradient-Push algorithm is a multi-agent analog of classical gradient descent. It should be noted that all of these algorithms are \emph{synchronous} in nature.

\subsubsection*{Asynchronous Multi-Agent Optimization}
The seminal work on asynchronous distributed optimization of Tsitsiklis \textit{et al.}~\cite{tsitsiklis1986distributed} considers the case where each agent holds one component of the optimization variable (or the entire optimization variable), and can locally evaluate a descent direction with respect to the global objective. Convergence is proved for a distributed gradient algorithm in that setting. However that setting is inherently different from the proposed problem formulation in which each agent does not necessarily have access to the global objective. Li and Basar~\cite{li1987asymptotic} study distributed asynchronous algorithms and prove convergence and asymptotic agreement in a stochastic setting, but assume a similar computation model to that of Tsitsiklis \textit{et al.}~\cite{tsitsiklis1986distributed} in which each agent updates a portion of the parameter vector using an operator which produces contractions with respect to the global objective.

Recently, several asynchronous multi-agent optimization methods have been proposed, such as:~\cite{wu2016decentralized}, which requires doubly-stochastic consensus over undirected graphs;~\cite{eisen2017decentralized,Lian2018asynchronous}, which require push-pull consensus over undirected graphs; and~\cite{mansoori2017superlinearly}, which assumes a model of asynchrony in which agents become activated according to a Poisson point process, and an active agent finishes its update before the next agent becomes activated.
In general, many of the asynchronous multi-agent optimization algorithms in the literature make restrictive assumptions regarding the nature of the agent updates ($\eg$, sparse Poisson point process~\cite{mansoori2017superlinearly}, randomized single activation~\cite{boyd2006randomized, dimakis2010gossip}, randomized multi-activation~\cite{nedic2011asynchronous, iutzeler2013asynchronous, wei20131, bianchi2016coordinate}).

\subsection{Contributions and Paper Organization}

We study an asynchronous implementation of the Subgradient-Push algorithm. Since we focus on problems with continuously differentiable objectives, we refer to the method as \emph{asynchronous Gradient-Push} (AGP). This paper draws motivation from our previous work~\cite{assran2017empirical} in which we empirically studied AGP and observed that it converges faster than state-of-the-art synchronous multi-agent algorithms. In this paper we provide theoretical convergence guarantees: when the local objective functions are strongly convex with Lipschitz-continuous gradients, we show that the iterates at each agent achieve consensus and converge to a neighborhood of the global minimum, where the size of the neighborhood depends on the degree of asynchrony. We consider a model of asynchrony which allows for heterogenous, bounded computation delays and communication delays. When the agents work at the same rate, convergence to the global minimizer is achieved. Moreover, if agents have knowledge of one another's potentially time-varying update rates, then they can set their step-sizes to achieve convergence to the global minimizer. In general, we relate the asymptotic worst-case error to the degree of asynchrony, as quantified by a bound on the delay. Agents do not need to know the delay bounds to execute the algorithm; the bounds only appear in the analysis.

Our analysis is based on several novel aspects: whereas previous work has used graph augmentation to model communication delays in consensus algorithms, here we augment with virtual nodes to model the effects of both communication \emph{and} computation delays on message passing in optimization algorithms.
Combining the graph augmentation with a (possibly time-varying) binary-valued activation function that is unique to each agent and directly multiplies its step-size, we are able to model the effect of heterogeneous update rates on the optimization procedure.
In contrast to previous work that makes additional assumptions on the agents' update rates, our problem formulation only assumes that the time-interval between an agents' consecutive activations is bounded.
Specifically, this formulation readily allows us to characterize the limit point as a \emph{deterministic function} of the agents' update rates, and to bound the rate of convergence when running AGP with constant or diminishing step-sizes.
Since synchronous gradient-push is a special case of AGP (with zero communication delay and unit computation delays), we obtain the first theoretical convergence guarantees for gradient-push with constant step-size.

We also develop peripheral results concerning an asynchronous version of the Push-Sum protocol used for consensus averaging that may be of independent interest. In particular, we show that agents running the Push-Sum protocol asynchronously converge to the average of the network geometrically fast, even in the presence of exogenous perturbations at each agent, where the constant of geometric convergence depends on the consensus-matrices' degree of ergodicity~\cite{hajnal1958weak} and a measure of asynchrony in the network.

In Sec.~\ref{sec:system-model} we describe the model of asynchrony considered in this paper. In Sec.~\ref{sec:asynchronous_push_sum} we expound the Asynchronous Perturbed Push-Sum consensus averaging protocol and give the associated convergence results. In Sec.~\ref{sec:asynchronous_subgradient_push} we formally describe the AGP optimization algorithm and present our main convergence results for both the constant and diminishing step-size cases. Sec.~\ref{sec:analysis} is devoted to the proof of the main results, and in Sec.~\ref{sec:experiments} we report numerical experiments on a high performance computing cluster. Finally, in Sec.~\ref{sec:conclusion}, we conclude and discuss extensions for future work.

\section{System Model}
\label{sec:system-model}

\subsection{Communication}
The multi-agent communication topology is represented by a directed graph $\Graph(\Vertices, \Edges)$, where
\begin{align*}
	\Vertices \defeq& \thesetof{v_i}{ i = 1, \ldots, n}, \\
	\Edges \defeq& \thesetof{ \edge{i}{j} }{ v_i \text{ can send messages to}\ v_j },
\end{align*}
are the set of agents and edges respectively. We refer to $\Graph(\Vertices, \Edges)$ as the \emph{reference graph} for reasons that will become apparent when we augment the graph with virtual agents. Let
\begin{align*}
	\Nin_j &\defeq \mathbf{card} \left(\thesetof{v_i}{ \edge{i}{j} \in \Edges} \right) \\
	\Nout_j &\defeq \mathbf{card} \left(\thesetof{v_i}{ \edge{j}{i} \in \Edges} \right)
\end{align*}
denote the cardinality of the \emph{in-neighbor} set and \emph{out-neighbor} set of agent $v_j$, respectively; we adopt the convention that $\edge{i}{i} \in \Edges$ for all $i$, \ie, every agent can send messages to itself.

\subsection{Discrete event sequence}

Without any loss of generality we can describe and analyze asynchronous algorithms as discrete sequences since all events of interest, such as message transmissions/receptions and local variable updates, may be indexed by a discrete-time variable~\cite{tsitsiklis1986distributed}.
We adopt notation and terminology for analyzing asynchronous algorithms similar to that developed in~\cite{tsitsiklis1986distributed}. Let $t[0] \in \R_{+}$ denote the time at which the agents begin optimization. We assume that there is a set of times $T = \{t[1], t[2], t[3], \ldots, \}$ at which one or more agents become \emph{activated}; i.e., completes a \textbf{Local Computation} and begins \textbf{Asynchronous Gossip}.
Let $T_i \subseteq T$ denote the subset of times at which agent $v_i$ in particular becomes activated.
Let $\Gset{k} \defeq \thesetof{v_i}{ t[k] \in T_i}$ denote the \emph{activation set} at time-index $k \in \mathbb{N}$, which is the set of agents that are activated at time $t[k]$.
For convenience, we also define the functions $\pi_i(k) \defeq \max \thesetof{k^\prime \in \mathbb{N}}{k^\prime < k,\ v_i \in \Gset{k^\prime}}$ for all $i$, which return the most recent time-index --- up to, but not including, time-index $k$ --- when agent $v_i$ was in the activation set.\footnote{To handle the corner-case at $k = 1$, we let $\pi_i(1)$ equal $0$ for all $i$.}

\subsection{Delays}
Recall that $t[k] \in T_i$ denotes a time at which agent $v_i$ becomes \emph{activated}: it completes a \textbf{Local Computation} (i.e., performs an update) and begins \textbf{Asynchronous Gossip} (i.e., sends a message to its neighbours by copying the outgoing message into its local send-buffer). For analysis purposes, messages are sent with an \emph{effective delay} such that they arrive right when the agent is ready to process the messages. That is, a message that is sent at time $t[k]$ and processed by the receiving agent at time $t[k^\prime]$, where $k' > k$, is treated as having experienced a time delay $t[k^\prime] - t[k]$ for the purpose of analysis, or equivalently a time-index delay $k^\prime - k$, even if the message actually arrives before $t[k^\prime]$ and waits in the receive-buffer.

Let $\taurate{k}{i} \defeq k - \pi_i(k)$ (defined for all $k$ such that $t[k] \in T_i$) denote the time-index processing delay experienced by agent $v_i$ at time $t[k]$. In words, if agent $v_i$ performs an update at some time $t[k]$, then it performed its last update at time $t[k - \taurate{k}{i}]$.
We assume that there exists a constant $\tauratemax < \infty$ independent of $i$ and $k$ such that $1 \leq \taurate{k}{i} \leq \tauratemax$.

Similarly, let $\tautran{k}{j}{i}$ (defined for all $k$ such that $t[k] \in T_j$) denote the time-index message delay experienced by a message sent from agent $v_i$ to agent $v_j$ at time $t[k]$. In words, if agent $v_i$ sends a message to agent $v_j$ at time $t[k]$, then agent $v_j$ will begin processing that message at time $t[k + \tautran{k}{j}{i}]$. We assume that there exists a constant $\tautranmax < \infty$ independent of $i$, $j$, and $k$, such that $\tautran{k}{j}{i} \leq \tautranmax$. In addition, we use the convention that $\tautran{k}{i}{i}= 0$ for all $i$ and $k \in \mathbb{N}$, meaning that agents always have immediate access to their most recent local variables. Thus $0 \leq \tautran{k}{j}{i} \leq \tautranmax$.

Since all agents enter the activation set (\ie, complete an update and initiate a message transmission to all their out-neighbors) at least once every $\tauratemax - 1$ time-indices, and because all messages are processed within at most $\tautranmax$ time-indices from when they are sent, it follows that each agent is guaranteed to process at least one message from each of its in-neighbors every $\taumax \defeq \tautranmax + \tauratemax - 1$ time-indices.

\subsection{Augmented Graph}
To analyze the AGP optimization algorithm we augment the reference graph by adding $\tautranmax$ virtual agents for each non-virtual agent. Similar graph augmentations have been used in~\cite{charalambous2015distributed, hadjicostis2014average} for synchronous averaging with transmission delays. One novel aspect of the augmentation described here is the use of virtual agents to model the effects of computation delays on message passing. To state the procedure concisely: for each non-virtual agent, $v_j$, we add $\tautranmax$ virtual agents, $v^{(1)}_j, v^{(2)}_j, \ldots, v^{(\tautranmax)}_j$, where each $v^{(r)}_j$ contains the messages to be received by agent $v_j$ in $r$ time-indices. As an aside, we may interchangeably refer to the non-virtual agents, $v_j$, as $v^{(0)}_j$ for the purpose of notational consistency. The virtual agents associated with agent $v_j$ are daisy-chained together with edges $\vedge{j}{r}{j}{r - 1}$, such that at each time-index $k$, and for all $r = 1, 2, \ldots, \tautranmax$, agent $v^{(r)}_j$ forwards its summed messages to agent $v^{(r - 1)}_j$. In addition, for each edge $\vedge{i}{0}{j}{0}$ in the reference graph (where $j \neq i$), we add the edges $\vedge{i}{0}{j}{r}$ in the augmented graph. This augmented model simplifies the subsequent analysis by enabling agent $v_i$ to send a message to agent $v^{(r)}_j$ with delay zero, rather than send a message to agent $v_j$ with delay $r$.\footnote{It is worth pointing out that we have not changed our definitions for the edge and vertex sets $\Edges$ and $\Vertices$ respectively; they are still solely defined in-terms of the non-virtual agents.} See Fig.~\ref{fig:aug_graph} for an example.

\begin{figure}[!t]
\centering

\tikzstyle{edge} = [very thick, -{Latex}]
\tikzstyle{vedge} = [very thick, cyan, dashdotted, -{Latex}]
\tikzstyle{non-virtualnode} = [draw, circle, fill=white, very thick]
\tikzstyle{virtualnode} = [draw, circle, cyan, fill=cyan!5, very thick, dashdotted]

\resizebox{0.6\columnwidth}{!}{
\begin{tikzpicture}

\node[non-virtualnode] (topcircle) {$v_1$};
\node[non-virtualnode] (leftcircle) [below left = 1cm of topcircle] {$v_4$};
\node[non-virtualnode] (rightcircle) [below right = 1cm of topcircle] {$v_2$};
\node[non-virtualnode] (bottomcircle) [below right = 1cm of leftcircle] {$v_3$};

\node[virtualnode] (v_topcircle1) [above left = 0.5cm of topcircle] {\footnotesize{$v^{(1)}_1$}};
\node[virtualnode] (v_topcircle2) [left = 0.5cm of v_topcircle1] {\footnotesize{$v^{(2)}_1$}};
\node[virtualnode] (v_rightcircle1) [above right = 0.5cm of rightcircle] {\footnotesize{$v^{(1)}_2$}};
\node[virtualnode] (v_rightcircle2) [above = 0.5cm of v_rightcircle1] {\footnotesize{$v^{(2)}_2$}};
\node[virtualnode] (v_bottomcircle1) [below right = 0.5cm of bottomcircle] {\footnotesize{$v^{(1)}_3$}};
\node[virtualnode] (v_bottomcircle2) [right = 0.5cm of v_bottomcircle1] {\footnotesize{$v^{(2)}_3$}};
\node[virtualnode] (v_leftcircle1) [below left = 0.5cm of leftcircle] {\footnotesize{$v^{(1)}_4$}};
\node[virtualnode] (v_leftcircle2) [below = 0.5cm of v_leftcircle1] {\footnotesize{$v^{(2)}_4$}};

\draw[edge] (leftcircle) -- (topcircle);
\draw[edge] (bottomcircle) -- (leftcircle);
\draw[edge] (rightcircle) -- (bottomcircle);
\draw[edge] (topcircle) -- (rightcircle);
\draw[edge] (bottomcircle) -- (topcircle);

\draw[edge] (topcircle) to[loop above, looseness=15] (topcircle);
\draw[edge] (leftcircle) to[loop left, looseness=15] (leftcircle);
\draw[edge] (rightcircle) to[loop right, looseness=15] (rightcircle);
\draw[edge] (bottomcircle) to[loop below, looseness=15] (bottomcircle);

\draw[vedge] (leftcircle) to[bend left=10] (v_topcircle1);
\draw[vedge] (leftcircle) to[bend left=15] (v_topcircle2);
\draw[vedge] (bottomcircle) to[bend right=5] (v_topcircle1);
\draw[vedge] (bottomcircle) to[bend right=5] (v_topcircle2);
\draw[vedge] (v_topcircle1) to[bend left=5] (topcircle);
\draw[vedge] (v_topcircle2) to[bend left=5] (v_topcircle1.west);
\draw[vedge] (topcircle) to[bend left=10] (v_rightcircle1);
\draw[vedge] (topcircle) to[bend left=20] (v_rightcircle2);
\draw[vedge] (v_rightcircle1) to[bend left=5] (rightcircle);
\draw[vedge] (v_rightcircle2) to[bend left=5] (v_rightcircle1);
\draw[vedge] (rightcircle) to[bend left=10] (v_bottomcircle1);
\draw[vedge] (rightcircle) to[bend left=15] (v_bottomcircle2);
\draw[vedge] (v_bottomcircle1) to[bend left=5] (bottomcircle);
\draw[vedge] (v_bottomcircle2) to[bend left=5] (v_bottomcircle1);
\draw[vedge] (bottomcircle) to[bend left=10] (v_leftcircle1);
\draw[vedge] (bottomcircle) to[bend left=15] (v_leftcircle2);
\draw[vedge] (v_leftcircle1) to[bend left=5] (leftcircle);
\draw[vedge] (v_leftcircle2) to[bend left=5] (v_leftcircle1);
\end{tikzpicture}
} 
\caption{Sample augmented graph of a $4$-agent reference network with a maximum time-index message transmission delay of $\tautranmax = 2$ time-indices. Solid lines correspond to non-virtual agents and edges. Dashed lines correspond to virtual agents and edges.} \label{fig:aug_graph}
\end{figure}
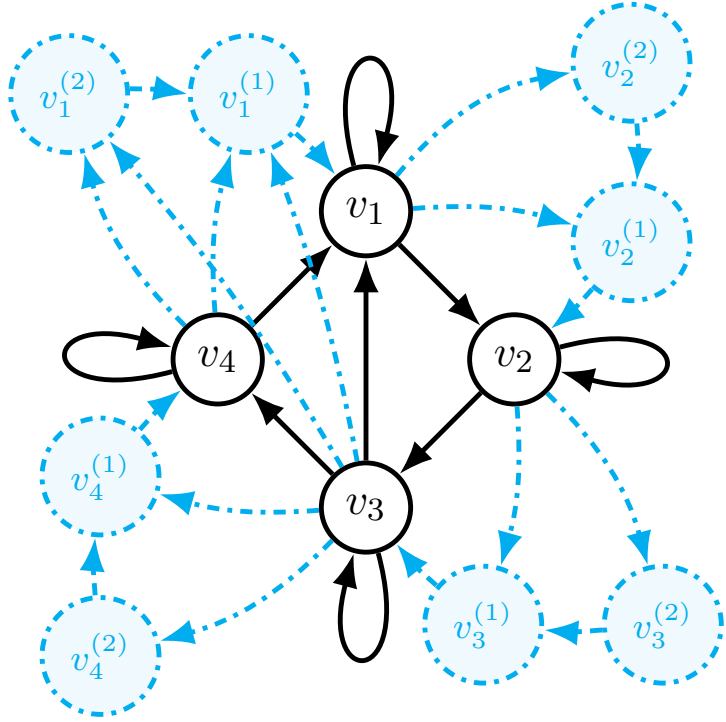

To adapt the augmented graph model for optimization we formulate the equivalent optimization problem
\begin{equation} \label{eq:aug_problem}
    \minimize\  \overline{F}(x) \defeq \sum^{\tautranmax}_{r=0} \sum_{i=1}^{n} f^{(r)}_i(x),
\end{equation}
where 
\[
f^{(r)}_i(x) = \begin{cases} f_i(x) &\text{ if } r=0,\\ 0 &\text{ otherwise.}\end{cases}
\]
In words, each of the non-virtual agents, $v^{(0)}_i$, maintains its original objective function $f_i( \cdot )$, and all the virtual agents are simply given the zero objective. Clearly $\overline{F}(x)$ defined in~\eqref{eq:aug_problem} is equal to $F(x)$ defined in~\eqref{eq:problem}.
We denote the state of a variable $x$ at time $t[k]$ with an augmented state matrix $\bm{x}[k] \in \R^{n(\tautranmax + 1) \times d}$
\begin{align}
\label{eq:notation}
\begin{split}
    \bm{x}[k] &\defeq 
        \begin{bmatrix}
            \bm{x}^{(0)}[k] \\
            \bm{x}^{(1)}[k] \\
            \vdots \\
            \bm{x}^{(\tautranmax)}[k]
        \end{bmatrix},
\end{split}
\end{align}
where each $\bm{x}^{(r)}[k] \in \R^{n \times d}$ is a block matrix that holds the copy of the variable $x$ at all the delay-$r$ agents in the augmented graph at time-index $k$.\footnote{In keeping with this notation, the block matrix $\bm{x}^{(0)}[k]$ corresponds to the non-virtual agents in the network.} More specifically, $x^{(r)}_i[k] \in \R^{d}$, the $i^{th}$ row of $\bm{x}^{(r)}[k]$, is the local copy of the variable $x$ held locally at agent $v^{(r)}_i$ at time-index $k$; below we generalize this notation for other variables as well.

For ease of exposition, we assume that the reference-graph is static and strongly-connected. The strongly-connected property of the directed graph is necessary to ensure that all agents are capable of influencing each other's values, and in Sec.~\ref{sec:conclusion} we describe how one can extend our analysis to account for time-varying directed communication-topologies.

\section{Asynchronous Perturbed Push-Sum}
\label{sec:asynchronous_push_sum}

Consensus-averaging is a fundamental building block of the proposed AGP algorithm. In this subsection we consider an asynchronous version of the synchronous Perturbed Push-Sum Protocol~\cite{nedich2015distributed}.
If we omit the gradient update in line~\ref{lst:gradient_update} of Algorithm~\ref{alg:agp}, then we recover the pseudocode for an asynchronous formulation of the Push-Sum consensus averaging protocol. Alternatively, if we replace the gradient term in line~\ref{lst:gradient_update} of Algorithm~\ref{alg:agp} with a general perturbation term, then we recover an asynchronous formulation of the Perturbed Push-Sum consensus averaging protocol.

\subsection{Formulation of Asynchronous (Perturbed) Push-Sum}
We describe the Asynchronous Perturbed Push-Sum algorithm in matrix form (which will facilitate analysis below) by stacking all of the agents' parameters at every update time into a parameter matrix using a similar notation to that in~\eqref{eq:notation}. The entire \textbf{Asynchronous Gossip} procedure can then be represented by multiplying the parameter-matrix by a so-called \emph{consensus-matrix} that conforms to the graph structure of the communication topology. The consensus matrices $\Pbar[k] \in \R^{n(\tautranmax + 1) \times n(\tautranmax  + 1)}$ for the augmented state model are defined as
\begin{equation}
\label{eq:augmented_consensus_matrix}
\Pbar[k] \defeq
\left[ \begin{array}{c c c c c}
    \Ptilde{0}[k] & \Eye{n} & \zeros & \cdots & \zeros \\[0.05in]
    \Ptilde{1}[k] & \zeros & \Eye{n}  & \cdots & \zeros \\[0.05in]
    \vdots & \vdots & \vdots  & \ddots & \vdots \\[0.05in]
    \Ptilde{\taumax - 1}[k] & \zeros & \zeros & \cdots & \Eye{n} \\[0.05in]
    \Ptilde{\tautranmax}[k] & \zeros & \zeros & \cdots & \zeros
\end{array} \right],
\end{equation}
where each $\Ptilde{r}[k] \in \R^{n \times n}$ is a block matrix defined as
\begin{equation}
\label{eq:augmented_consensus_matrix_subblock}
\left[ \Ptilde{r}[k] \right]_{ji} \defeq
\begin{cases}
    \frac{1}{\Nout_i}, & v_i \in \Gset{k},\ (j, i) \in \Edges,\ \text{and}\ \tautran{k}{j}{i} = r, \\
    1, & v_i \notin \Gset{k},\ r=0,\ j=i, \\
    0, & \text{otherwise.}
\end{cases}
\end{equation}
In words, when a non-virtual agent is in the activation set, it sends a message to each of its out-neighbours in the reference graph with some arbitrary, but bounded, delay $r$.
When a non-virtual agent is not in the activation set, it keeps its value and does not gossip. Furthermore, since we have chosen a convention in which messages between agents are sent with some effective message delay, $\tautran{k}{j}{i}$, it follows that non-virtual agents do not receive any new messages while outside the activation set. Virtual agents, on the other hand, simply forward all of their messages to the next agent in the delay chain at all time-indices $k$, and so there is no notion of virtual agents belonging to (or not belonging to) the activation set. The activation set is exclusively a construct for the non-virtual agents. Observe that, by definition, the matrices $\Pbar[k]$ are column stochastic at all time-indices $k$.

\begin{figure}[!t]
\begin{algorithm}[H]
\caption{Asynchronous Perturbed Push-Sum Averaging}
\label{eq:asynch_pert}
\begin{algorithmic}
\For{k = 0, 1, 2, \ldots}{\text{termination}}
\begin{align}
        \bm{x}[k +1] &= \Pbar[k] \left( \bm{x}[k] + \bm{\eta}[k] \right) \label{eq:pert_itr_1} \\
        y[k + 1] &= \Pbar[k] y[k] \label{eq:pert_itr_2} \\
       \bm{z}[k + 1] &= \text{diag}(y[k + 1])^{-1}\bm{x}[k + 1] \label{eq:pert_itr_3}
\end{align}
\end{algorithmic}
\end{algorithm}
\end{figure}

To analyze the Asynchronous Perturbed Push-Sum averaging algorithm from a global perspective, we use the matrix-based formulation provided in Algorithm~\ref{eq:asynch_pert}, where $\bm{\eta}[k] \in \R^{n(\tautranmax + 1) \times d}$ is a perturbation term, and the matrices $\Pbar[k]$ are as defined in~\eqref{eq:augmented_consensus_matrix} for the augmented state, and $\bm{x}[k]$, $y[k]$, and $\bm{z}[k]$ are also defined with respect to the augmented state. At all time-indices $k$, each agent $v^{(r)}_i$ locally maintains the variables $x^{(r)}_i[k], z^{(r)}_i[k], \in \R^d$, and $y^{(r)}_i[k] \in \R$. The non-virtual agent initializations are $x^{(0)}_i[0] \in \R^d$, and $y^{(0)}_i[0] = 1$. The virtual agent initializations are $x^{(r)}_i[0] = \zeros$, and $y^{(r)}_i[0] = 0$ (for all $r \neq 0$).\footnote{Note, given the initializations, the virtual agents could potentially have $z^{(r)}_i[k +1] = 0/0$ (division by zero) in update equation~\eqref{eq:asynch_itr_3}, but this is a non-issue since $z^{(r)}_i$ (for all $r \neq 0$) is never used.} This matrix-based formulation describes how the agents' values evolve at those times $t[k + 1] \in T = \{t[1], t[2], t[3], \ldots, \}$ when one or more agents complete an update, which in this case consists of summing received messages. The time-varying consensus-matrices $\Pbar[ \cdot ]$ capture the asynchronous delay-prone communication dynamics.

\subsection{Main Results for Asynchronous (Perturbed) Push-Sum}

In this subsection we present the main convergence results for the Asynchronous (Perturbed) Push-Sum consensus averaging protocol. We briefly describe some notation in order to state the main results. Let $\Nout_{max} \defeq \max_{1 \leq j \leq n} \Nout_j$ represent the maximum number of out-neighbors associated to any non-virtual agent. Let $\xbar[k] \defeq \ones^{\T} \bm{x}[k] / n$ be the mutual time-wise average of the variable $x$ at time-index $k$. Let the scalar $\psi$ represent the number of possible types (zero/non-zero structures) that an $n(\tautranmax + 1) \times n(\tautranmax + 1)$ stochastic, indecomposable, and aperiodic (SIA) matrix can take (hence $\psi < {2^{(n(\tautranmax + 1) )}}^2$).\footnote{See~\cite{wolfowitz1963products} for a definition of SIA matrices.} Let the scalar $\lambda > 0$ represent the maximum Hajnal Coefficient of Ergodicity~\cite{hajnal1958weak} taken over the product of all possible $(\taumax + 1)(\psi + 1)$ consecutive consensus-matrix products:
\[
    \lambda \defeq \max_{\bm{A}} \left( 1 - \min_{j_1, j_2} \sum_{i} \min \left \{ \left[\bm{A} \right]_{i, j_1}, \left[ \bm{A} \right]_{i, j_2} \right \} \right),
\]
such that
\[
	\bm{A} \in \thesetof{ \Pbar[k + (\taumax + 1)(\psi + 1)] \cdots \Pbar[k + 2] \Pbar[k + 1]}{ k \geq 0 },
\]
where $\taumax \defeq \tautranmax + \tauratemax - 1$.
We prove that $\lambda$ is strictly less than $1$ and guaranteed to exist. Let $\delta_{min}$ represent a lower bound on the entries in the first $n$-rows of the product of $n(\taumax + 1)$ or more consecutive consensus-matrices (rows only corresponding to the non-virtual agents):
\[
	\delta_{min} \defeq \min_{\substack{i, j, k, \ell}} \left[ \Pbar[k + \ell] \cdots \Pbar[k + 2] \Pbar[k + 1] \right]_{i, j},
\]
where the $\min$ is taken over all $i = 1, 2, \ldots, n$, $j = 1, 2, \ldots, n(\tauratemax + 1)$, $k \geq 0$, and $\ell \geq n(\tauratemax + 1)$.

\begin{assumption}[Communicability]
\label{ass:base_comm}
All agents influence each other's values sufficiently often, in particular:
\begin{enumerate}
    \item The reference graph $\Graph(\Vertices, \Edges)$ is static and strongly connected.
    \item The communication and computation delays are bounded: $\tautranmax < \infty$ and $\tauratemax < \infty$.
\end{enumerate}
\end{assumption}

\begin{theorem}[Convergence Rate of Asynchronous Perturbed Push-Sum Averaging]
\label{th:avg_rate}
Suppose that Assumption~\ref{ass:base_comm} is satisfied. Then it holds for all $i = 1, 2, \ldots, n$, and $k \geq 0$, that
\begin{align*}
    \norm{ z^{(0)}_i[k] - \xbar[k] }_1 \leq& C q^k \norm{ x^{(0)}_i[0]}_{1} + C \sum^k_{s=0}q^{k - s} \norm{ \eta_i[s] }_{1},
\end{align*}
where $q \in (0,1)$ and $C > 0$ are given by
\begin{align*}
	q = \lambda^\frac{1}{{(\psi + 1)(\taumax + 1)}}, \quad \text{and} \quad C < \frac{2}{\lambda^{\frac{\psi + 2}{\psi + 1}} \delta_{min}} \approx \frac{2}{\lambda \delta_{min}},
\end{align*}
and $\delta_{min} = \left( \frac{1}{\Nout_{max}} \right)^{n(\taumax + 1)}$.
\end{theorem}

\begin{remarkth*}
To adapt the proof to $B$-strongly connected time-varying directed graphs, one would instead define $\lambda$ as the maximum Hajnal Coefficient of Ergodicity~\cite{hajnal1958weak} taken over the product of all possible $(\taumax + 1 + B)(\psi + 1)$  consecutive matrix products (instead of all $(\taumax + 1)(\psi + 1)$ consecutive matrix products).
A sufficient assumption in order to prove that $\lambda < 1$ is that a message in transit does not get dropped when the graph topology changes.
\end{remarkth*}

\begin{corollary}[Convergence to a Neighbourhood for Non-Diminishing Perturbation]
\label{cor:avg_rate_neighbourhood}
Suppose that the perturbation term is bounded for all $i = 1, 2, \ldots, n$; i.e., there exists a positive constant $L < \infty$ such that
\[
    \norm{\eta_i[k]}_1 \leq L, \quad \text{for all } i = 1, 2, \ldots, n .
\]
Then, for all $i = 1, 2, \ldots, n$,
\[
    \lim_{k \to \infty} \norm{z^{(0)}_i[k] - \xbar[k] }_1 \leq \frac{{C} L}{1 - q}.
\]
\end{corollary}

\begin{remark}
\label{rem:ram2010distributed}
From~\cite[Lemma 3.1]{ram2010distributed} we know that if $q \in (0, 1)$, and $\lim_{s \to \infty} \alpha[s] = 0$, then
\[
   \lim_{k \to \infty} \sum^{k}_{s = 0} q^{{k} - s} \alpha[s] = 0.
\]
\end{remark}

\begin{corollary}[Exact Convergence for Vanishing Perturbation]
\label{cor:avg_rate_vanish}
Suppose that the perturbation term vanishes as $k$ (the time-index) tends to infinity, i.e.,
\[
    \lim_{k \to \infty} \norm{\bm{\eta[k]}}_1 = 0,
\]
then from the result of Theorem~\ref{th:avg_rate} and Remark~\ref{rem:ram2010distributed}, it holds for all $i = 1, 2, \ldots, n$ that
\[
    \lim_{k \to \infty} \norm{ z^{(0)}_i[k] - \xbar[k] }_1 = 0.
\]
\end{corollary}

\begin{corollary}[Geometric Convergence of Asynchronous (Unperturbed) Push-Sum Averaging]
\label{cor:avg_rate_vanish}
Suppose that for all $i=1,2, \ldots, n$, and $k\geq 0$, it holds that $\eta_i[k] = \zeros$.
Then from the result of Theorem~\ref{th:avg_rate}, it holds for all $i = 1, 2, \ldots, n$, and $k\geq 0$ that
\[
	\norm{ z^{(0)}_i[k] - \xbar[0] }_1 \leq C q^k \norm{x^{(0)}_i[0]}_1.
\]
\end{corollary}

The proof of Theorem~\ref{th:avg_rate} is omitted and can be found in~\cite{Assran2018thesis}. In brief, the asymptotic product of the asynchronous consensus-matrices, $\Pbar[k] \cdots \Pbar[1] \Pbar[0]$ (for sufficiently large $k$) is SIA, and furthermore, the entries in the first $n$ rows of the asymptotic product (corresponding to the non-virtual agents) are bounded below by a strictly positive quantity. Applying standard tools from the literature concerning SIA matrices~\cite{wolfowitz1963products} we show that the columns of the asymptotic product of consensus-matrices weakly converge to a stochastic vector sequence at a geometric rate. Substituting this geometric bound into the definition of the asynchronous perturbed Push-Sum updates in Algorithm~\ref{eq:asynch_pert}, and after algebraic manipulation similar to that in~\cite{nedich2015distributed} (which analyzes synchronous delay-free Perturbed Push-Sum), we obtain the desired result.

\section{Asynchronous Gradient-Push}
\label{sec:asynchronous_subgradient_push}

In this section we expound the proposed AGP optimization and present our main convergence results. Our model of asynchrony implies that agents may gossip at different rates, may communicate with arbitrary transmission delays, and may perform gradient steps with stale (outdated) information.

\subsection{Formulation of Asynchronous Gradient-Push}
To analyze the AGP optimization algorithm from a global perspective, we use the matrix-based formulation provided in Algorithm~\ref{eq:asynch_opt}. At all time-indices $k$, each agent $v^{(r)}_i$ locally maintains the variables $x^{(r)}_i[k], z^{(r)}_i[k] \in \R^d$, and $y^{(r)}_i[k] \in \R_+$.
The \emph{non-virtual} agents initialize these to $z^{(0)}_i[0] = x^{(0)}_i[0] \in \R^d$, and $y^{(0)}_i[0] = 1$. The \emph{virtual} agents' variables are initialized to $z^{(r)}_i[0] = x^{(r)}_i[0] = \zeros$, and $y^{(r)}_i[0] = 0$ for all $r \neq 0$.
This matrix-based formulation describes how the agents' values evolve at those times $t[k + 1] \in T = \{t[1], t[2], t[3], \ldots, \}$ when one or more agent becomes activated (completes an update).
The asynchronous delay-prone communication dynamics are accounted for in the consensus-matrices $\Pbar[\cdot]$, and the matrix-valued function $ \Grad{k + 1} \in \R^{n(\tautranmax + 1) \times d}$ is defined as
\begin{align*}
    \Grad{k + 1} &\defeq
        \begin{bmatrix}
             \nabla \bm{f}^{(0)}(\bm{z}^{(0)}[k + 1]) \\
             \zeros \\
             \vdots \\
             \zeros
         \end{bmatrix},
\end{align*}
where $\nabla \bm{f}^{(0)}(\bm{z}^{(0)}[k + 1]) \in \R^{n \times d}$ denotes a block matrix with its $i^{th}$ row equal to
\[
	\alpha_i[k + 1] \delta_i[k + 1] \nabla f^{(0)}_i(z^{(0)}_i[k + 1]).
\]
The scalar $\alpha_i[k + 1]$ denotes node $v_i$'s local step-size.
The scalar $\delta_i[ \cdot ]$ is equal to $1$ when agent $v_i$ is activated, and equal to $0$ otherwise.
Recall that agents can only update their local step-sizes when they are activated (i.e., they complete a local gradient step, cf.~Algorithm~\ref{alg:agp}).
Therefore, if agent $v_i$ is \emph{not} activated at time-index $k$, then $\alpha_i[k]$ is equal to $\alpha_i[ \pi_i(k) ]$, the agent's most recently used step-size.\footnote{Note: if an agent is not activated at time-index $k$, then its step-size at that time does have any effect on the execution of the algorithm. We introduce this convention here simply so that the step-size value is well-defined at all times.}

\begin{algorithm}[t]
\caption{Asynchronous Gradient Push Optimization}
\label{eq:asynch_opt}
\begin{algorithmic}
\For{k = 0, 1, 2, \ldots}{\text{termination}}
\begin{align}
        \bm{x}[k +1] &= \Pbar[k] \left( \bm{x}[k] - \Grad{k} \right) \label{eq:asynch_itr_1} \\
        y[k + 1] &= \Pbar[k] y[k] \label{eq:asynch_itr_2} \\
        \bm{z}[k + 1] &= \text{diag}(y[k + 1])^{-1}\bm{x}[k + 1] \label{eq:asynch_itr_3}
\end{align}
\end{algorithmic}
\end{algorithm}

\subsection{Main results for Asynchronous Gradient-Push}
In this subsection we present the main convergence results for the AGP algorithm.

\begin{assumption}[Existence, Convexity, and Smoothness]
\label{ass:base_obj}
Assume that:
\begin{enumerate}
    \item A minimizer of \eqref{eq:problem} exists; i.e., $\argmin_x F(x) \neq \emptyset$.
    \item Each function $f_i(x): \R^d \rightarrow \R$ is $\mu_i$-strongly convex, and has $M_i$-Lipschitz continuous gradients.
\end{enumerate}
\end{assumption}

Let $M \defeq \max_i M_i$ and $\mu \defeq \min_i \mu_i$ denote the global Lipschitz constant and modulus of strong convexity, respectively.
Let $\xstar \defeq \argmin \overline{F}(x)$ denote the global minimizer, and let $\xstar_i \defeq \argmin f_i(x)$ denote the minimizer of node $v_i$'s local objective.

\begin{assumption}[Step-Size Bound]
\label{ass:step_size_bound}
Assume that for all agents $v_i$, the terms in the step-size sequence $\{ \alpha_i[k] \}$ satisfy
\[
	\alpha_i[k] \leq \frac{\mu}{2 M^2} \left( \frac{1}{\Nout_{max}} \right)^{n(\taumax + 1)} \quad \quad \forall k \in \mathbb{N}.
\]
\end{assumption}

\begin{theorem}[Bounded Iterates and Gradients]
\label{th:bounded_iterates_gradients}
Suppose Assumptions~\ref{ass:base_obj} and~\ref{ass:step_size_bound} are satisfied. Then there exist finite constants $L, D > 0$ such that,
\begin{align*}
    \sup_{k} \norm{\nabla f_i(z_i[k])} \leq L, \quad \sup_k \norm{ \xbar[k] } < D.
\end{align*}
\end{theorem}
The proof of Theorem~\ref{th:bounded_iterates_gradients} appears in~\cite{Assran2018thesis}.
Next we state our main results, the proofs of which all appear in Sec.~\ref{sec:analysis}. When nodes run asynchronously and at different rates, AGP may not converge precisely to the solution $\xstar$ of \eqref{eq:problem}.

\begin{definition}[Re-weighted objective]
\label{def:reweighted_obj}
Suppose Algorithm~\ref{alg:agp} is run from time $t[0]$ up to time $t[K]$ for some integer $K > 0$.
For all $i \in [n]$, let
\begin{equation}
	p^{(K)}_i \defeq \sum^{K - 1}_{k=0} \alpha_i[k] \delta_i[k], \quad \text{and} \quad \overline{p}^{(K)}_i \defeq \frac{p^{(K)}_i}{ \sum^n_{i=1} p^{(K)}_i }.
\end{equation}
Define the re-weighted objective
\begin{equation}
	F_K( \cdot ) \defeq \sum^n_{i=1}  \overline{p}^{(K)}_i  f_i( \cdot ), \label{eqn:F_K}
\end{equation}
and let $\xstar_K$ denote the minimizer of $F_K( \cdot )$.
\end{definition}

We can characterize how far $\xstar_K$ may be from $\xstar$. Let $\kappa \defeq M/\mu$ denote the condition number of the global objective $F(x)$, let $\xstar_i$ denote the minimizer of $f_i(x)$, let $S_i \defeq \norm{ \xstar_i - \xstar}$, let $S_{i,j} \defeq \norm{ \xstar_i - \xstar_j }$ denote the pairwise distance of agent $v_i$'s minimizer to agent $v_j$'s minimizer, and let $\overline{S} \defeq \max_{i \in [n]} \min_{j \in [n]} ( S_{i,j} + S_{j} )$.

\begin{theorem}[Bound on Distance of Minimizers]
\label{th:bound_dist_minimziers}
Suppose Algorithm~\ref{alg:agp} is run from time $t[0]$ up to time $t[K]$, for some integer $K > 0$.
Let 
\[
\Delta^{(K)} \defeq \sqrt{ \sum^n_{i=1} \abs{\frac{1}{n} - \overline{p}^{(K)}_i} }.
\]
If Assumption~\ref{ass:base_obj} holds, then
\[
	\norm{ \xstar_K - \xstar } \leq \frac{\overline{S} \sqrt{\kappa} \;\Delta^{(K)}}{\sqrt{2}} ,
\]
where $\overline{p}^{(K)}_i \in (0,1)$ and $\xstar_K$ are defined in Definition~\ref{def:reweighted_obj}, and $\xstar$ is the minimizer of~\eqref{eq:problem}.
\end{theorem}
Theorem~\ref{th:bound_dist_minimziers} bounds the distance between the minimizer of the re-weighted objective (Definition~\ref{def:reweighted_obj}) and the minimizer of the original (unbiased) objective~\eqref{eq:problem}.
The bound depends on the condition number of the global objective, the pairwise distance between agents' local minimizers, the distance between agents' local minimizers and the global (unbiased) minimizer, and the degree of asynchrony in the network.
In particular, the quantity $\Delta^{(K)}$ denotes the bias introduced from the processing delays. If agents work at roughly the same rate, then $\Delta^{(K)}$ is close to $0$. On the other hand, if there is a large disparity between agents'~update rates, then $\Delta^{(K)}$ is close to $\sqrt{2}$.

\begin{assumption}[Constant Step-Size]
\label{ass:step_size_const}
Suppose Algorithm~\ref{alg:agp} is run from time $t[0]$ up to time $t[K]$, for some integer $K > 0$.
For a given $\theta \in (0,1)$, assume that there exist constants $B > 0$ and $w_i \geq 1$, for all $i \in [n]$, such that each agent $v_i$ sets its local step-size as
\[
	\alpha_i[k] \defeq \alpha_i = \frac{w_i B}{ K^\theta}.
\]
\end{assumption}
Note that Assumption~\ref{ass:step_size_const} prescribes a constant step-size. It reads: first fix the total number of iterations $K$, and then use $K$ to inform the choice of a constant step-size.\footnote{In practice it may be difficult to determine $K$ ahead of time, since $K$ is the total number of iterations/updates performed \emph{across the entire network}. However in some implementations it may be possible to maintain a (possibly approximate) global count of the number of iterations performed (e.g., by running a separate consensus algorithm in parallel) and use this as a stopping criterion.}

\begin{theorem}[Convergence of Asynchronous Gradient Push for Constant Step-Size]
\label{th:const_convergence_opt}
Suppose Algorithm~\ref{alg:agp} is run from time $t[0]$ up to time $t[K]$, for some integer $K > 0$, and suppose that Assumptions~\ref{ass:base_comm},~\ref{ass:base_obj},~\ref{ass:step_size_bound}, and~\ref{ass:step_size_const} hold. 
Then there exist finite positive constants $A_1$, $A_2$, and $A_3$ such that
\begin{align*}
	\frac{1}{K} \sum^{K -1}_{k=0} \norm{\xbar[k] - \xstar_K}^2 \leq& \frac{1}{ K^\theta} \left( \frac{n (A_1 + A_3)}{2 \mu B} \right) +  \frac{1}{ K} \left( \frac{n A_2}{2 \mu B} \right) \\
	& + \frac{1}{K^{1 - \theta}} \left(\frac{n \left( \norm{ \xbar[0] - \xstar_K}^2 \right)}{2 \mu B } \right),
\end{align*}
where $\theta \in (0, 1)$ is defined in Assumption~\ref{ass:step_size_const}, and $\xstar_K$ is the minimizer of the re-weighted objective defined in Definition~\ref{def:reweighted_obj}.
\end{theorem}

Explicit expressions for $A_1$, $A_2$, and $A_3$ are given in Lemma~\ref{lem:bounded_summations_const} below. Both $A_2$ and $A_3$ depend on $C$ and $q$, and hence on the delay bound $\taumax$.

\begin{corollary}[Convergence of Semi-Synchronous Gradient Push for Constant Step-Size]
\label{cor:const_semi_synch_convergence_opt}
Suppose the assumptions made in Theorem~\ref{th:const_convergence_opt} hold, and suppose that $\tauratemax = 1$ and each agent $v_i$ sets its local step-size scaling factor $w_i = 1$.
Then
\begin{align*}
	\frac{1}{K} \sum^{K -1}_{k=0} \norm{\xbar[k] - \xstar}^2 \leq& \frac{1}{ K^\theta} \left( \frac{n (A_1 + A_3)}{2 \mu B} \right) +  \frac{1}{ K} \left( \frac{n A_2}{2 \mu B} \right)\\
	& +  \frac{1}{K^{1 - \theta}} \left(\frac{n \left( \norm{ \xbar[0] - \xstar }^2 \right)}{2 \mu B } \right),
\end{align*}
where $\xstar$ is the minimizer of~\eqref{eq:problem}.
\end{corollary}
Corollary~\ref{cor:const_semi_synch_convergence_opt} states that if the agents perform gradient updates at the same rate, then they converge to the unbiased global minimizer, even in the presence of persistent, but bounded, message delays.

\begin{definition}[Local iteration counter]
For each agent $v_i$, and all integers $k \geq 0$, define the local iteration counter
\[
	c_i[k] \defeq \sum^{k}_{\ell = 0} \delta_i[\ell]
\]
to be the number of updates performed by agent $v_i$ in the time-interval $(t[0], t[k]]$.
By convention, for all $i \in [n]$, we take $\delta_i[0] \defeq 1$, and thus $c_i[0]= 1$.

\begin{corollary}[Convergence of Asynchronous Gradient Push for Known Update Rates]
\label{cor:const_congergence_opt_known_udpate_rates}
Suppose the assumptions made in Theorem~\ref{th:const_convergence_opt} hold, and suppose that each agent $v_i$ has prior knowledge of $c_i[K - 1]$, the number of local iterations it will have completed before time $t[K]$.
If each agent $v_i$ sets its local step-size scaling factor
\[
	w_i \defeq \frac{K}{c_i[K - 1]} \geq 1,
\]
then
\begin{align*}
	\frac{1}{K} \sum^{K -1}_{k=0} \norm{\xbar[k] - \xstar}^2 \leq& + \frac{1}{ K^\theta} \left( \frac{n (A_1 + A_3)}{2 \mu B} \right) +  \frac{1}{ K} \left( \frac{n A_2}{2 \mu B} \right) \\
	& + \frac{1}{K^{1 - \theta}} \left(\frac{n \left( \norm{ \xbar[0] - \xstar }^2 \right)}{2 \mu B } \right),
\end{align*}
where $\xstar$ is the minimizer of~\eqref{eq:problem}.
\end{corollary}
Corollary~\ref{cor:const_congergence_opt_known_udpate_rates} states that if the agents know one another's update rates, then they can set their step-sizes to guarantee convergence to the unbiased global minimizer, even in the presence of persistent, but bounded, processing and message delays.
In particular, slower agents can simply scale up their step-size to compensate for their slower update rates.

We also provide guarantees for a version of the algorithm using diminishing step sizes.

\end{definition}
\begin{assumption}[Step-Size Decay]
\label{ass:step_size_decay}
For a given $\theta \in (0.5,1)$, assume that there exist constants $B > 0$ and $w_i \geq 1$, for all $i \in [n]$, such that each agent $v_i$ sets its local step-size as
\[
	\alpha_i[k] \defeq \frac{w_i B}{ (c_i[k])^\theta}.
\]
\end{assumption}
\begin{remark}
\label{rem:step_size_decay}
Note that if Assumption~\ref{ass:step_size_decay} holds, then
\[
	\frac{B}{n (k + 1)^\theta} \leq \frac{1}{n} \sum^n_{i=1} \alpha_i[k] \delta_i[k] \leq  \frac{( \frac{1}{n} \sum^n_{i=1} w_i ) B (\tauratemax)^\theta }{ (k + \tauratemax)^\theta },
\]
where $\theta \in (0.5, 1)$ is defined in Assumption~\ref{ass:step_size_decay}
\end{remark}

\begin{theorem}[Convergence of Asynchronous Gradient Push for Diminishing Step-Size]
\label{th:convergence_opt}
Suppose Algorithm~\ref{alg:agp} is run from time $t[0]$ up to time $t[K]$, for some integer $K > 0$.
If Assumptions~\ref{ass:base_comm},~\ref{ass:base_obj},~\ref{ass:step_size_bound}, and~\ref{ass:step_size_decay} hold, then there exists a finite positive constant $A$ such that
\[
	\frac{1}{K} \sum^{K -1}_{k=0} \norm{\xbar[k] - \xstar_K}^2 \leq \frac{1}{K^{1 - \theta}} \left(\frac{n \left( \norm{ \xbar[0] - \xstar_K }^2 + A \right)}{2 \mu B } \right),
\]
where $\theta \in (0.5, 1)$ is defined in Assumption~\ref{ass:step_size_decay}, and $\xstar_K$ is the minimizer of the re-weighted objective defined in Definition~\ref{def:reweighted_obj}.
\end{theorem}

Theorem~\ref{th:convergence_opt} states that in the presence of persistent, but bounded, message and processing delays, the agents converge to the minimizer of a re-weighted version of the original problem, where the re-weighting values are completely determined by the agents' respective~cumulative step-sizes during the execution of the algorithm. The constant $A$ depends on the delay bound $\taumax$; see Lemma~\ref{lem:bounded_summations} below for more details.

\begin{corollary}[Exact Consensus for Asynchronous Gradient Push]
\label{cor:exact_consensus_dim}
Suppose the assumptions made in Theorem~\ref{th:convergence_opt} hold. Then, for all $i \in [n]$,
\[
	\lim_{k \to \infty} \norm{z_i[k] - \xbar[k]}  = 0.
\]
\end{corollary}
\begin{proof}
Notice that the Asynchronous Gradient Push updates in Algorithm~\eqref{eq:asynch_opt} can be regarded as Asynchronous Perturbed Push-Sum updates, with perturbation $\bm{\eta}[k]$ given by $- \Grad{k}$. Since the gradients remain bounded by Theorem~\ref{th:bounded_iterates_gradients}, and the local step-sizes go to zero by Assumption~\ref{ass:step_size_decay}, the conditions for Corollary~\ref{cor:avg_rate_vanish} are satisfied, and it follows that $\lim_{k \to \infty} \norm{ z_i[k] - \xbar[k] } = 0$.
\end{proof}
Corollary~\ref{cor:exact_consensus_dim} states that if all agents use a diminishing step-size, then they will achieve consensus, even in the presence of persistent, but bounded, processing and message delays.

\begin{corollary}[Convergence of Semi-Synchronous Gradient Push for Diminishing Step-Size]
\label{cor:semi_synch_convergence_opt}
Suppose the assumptions made in Theorem~\ref{th:convergence_opt} hold. If $\tauratemax = 1$ and each agent $v_i$ sets its local step-size scaling factor $w_i = 1$, then
\[
	\frac{1}{K} \sum^{K -1}_{k=0} \norm{\xbar[k] - \xstar}^2 \leq \frac{1}{K^{1 - \theta}} \left(\frac{n \left( \norm{ \xbar[0] - \xstar }^2 + A \right)}{2 \mu B } \right),
\]
where $\xstar$ is the minimizer of~\eqref{eq:problem}.
\end{corollary}
Corollary~\ref{cor:semi_synch_convergence_opt} states that if the agents perform gradient updates at the same rate, then they converge to the (unbiased) global minimizer, even in the presence of persistent, but bounded, message delays.

\begin{corollary}[Convergence of Asynchronous Gradient Push for Known Update Rates]
\label{cor:congergence_opt_known_udpate_rates}
Suppose the assumptions made in Theorem~\ref{th:convergence_opt} hold, and suppose that each agent $v_i$ has prior knowledge of $c_i[K - 1]$, the number of local iterations it will have completed before time $t[K]$.
If each agent $v_i$ sets its local step-size scaling factor
\[
	w_i \defeq \frac{ \left( \sum^{K - 1}_{k=0} \frac{1}{(k + 1)^\theta} \right) }{ \left( \sum^{c_i[K - 1] - 1}_{k=0} \frac{1}{(k + 1)^\theta} \right)} \geq 1
\]
for some $\theta \in (0.5, 1)$ (as per Assumption~\ref{ass:step_size_decay}), then
\[
	\frac{1}{K} \sum^{K -1}_{k=0} \norm{\xbar[k] - \xstar}^2 \leq \frac{1}{K^{1 - \theta}} \left(\frac{n \left( \norm{ \xbar[0] - \xstar }^2 + A \right)}{2 \mu B } \right),
\]
where $\xstar$ is the minimizer of~\eqref{eq:problem}.
\end{corollary}
Corollary~\ref{cor:congergence_opt_known_udpate_rates} states that if the agents know one another's update rates, then they can set their step-sizes to guarantee convergence to the ubiased global minimizer, even in the presence of persistent, but bounded, processing and message delays.
In particular, slower agents can simply scale up their step-size to compensate for their slower update rates.

\section{Analysis}
\label{sec:analysis}

\subsection{Proof of Theorem~\ref{th:bound_dist_minimziers}}

\begin{proof}[\unskip\nopunct]
Using the strong convexity of the global objective, we have
\begin{equation}
\label{eq:bound_dist_minimizer_eq1}
	\norm{ \xstar_K - \xstar }^2 \leq \frac{2}{\mu} \sum^n_{i=1} \frac{1}{n} (f_i(\xstar_K) - f_i(\xstar)),
\end{equation}
and
\begin{equation}
\label{eq:bound_dist_minimizer_eq2}
	\norm{ \xstar_K - \xstar }^2 \leq \frac{2}{\mu} \sum^n_{i=1} \overline{p}^{(K)}_i (f_i(\xstar) - f_i(\xstar_K)).
\end{equation}
Summing~\eqref{eq:bound_dist_minimizer_eq1} and~\eqref{eq:bound_dist_minimizer_eq2} and multiplying through by $1/2$, we obtain that
\begin{equation*}
	\norm{ \xstar_K - \xstar }^2 \leq \frac{1}{\mu} \sum^n_{i=1} \left( ( f_i(\xstar_K) - f_i(\xstar) ) \left(\frac{1}{n} -\overline{p}^{(K)}_i \right) \right).
\end{equation*}
Adding and subtracting $\frac{1}{\mu} f_i(\xstar_i)$, we have
\begin{align}
\label{eq:bound_dist_minimizer_eq3}
\begin{split}
	\norm{ \xstar_K - \xstar }^2 \leq& \frac{1}{\mu} \sum^n_{i=1} (f_i(\xstar_K) - f_i(\xstar_i)) \left(\frac{1}{n} - \overline{p}^{(K)}_i\right) \\
	&- \frac{1}{\mu} \sum^n_{i=1} (f_i(\xstar) - f_i(\xstar_i)) \left(\frac{1}{n} -\overline{p}^{(K)}_i \right).
\end{split}
\end{align}
Define the index set $\mathcal{I} \defeq \{ i \in [n] | \frac{1}{n} - \overline{p}^{(K)}_i  \geq 0 \}$, and its complement $\mathcal{I}^C \defeq \{ i \in [n] | \frac{1}{n} - \overline{p}^{(K)}_i  < 0 \}$.
We can further bound~\eqref{eq:bound_dist_minimizer_eq3} as
\begin{align}
\label{eq:bound_dist_minimizer_eq4}
\begin{split}
	\norm{ \xstar_K - \xstar }^2 \leq& \frac{1}{\mu} \sum_{i \in \mathcal{I}} (f_i(\xstar_K) - f_i(\xstar_i)) \abs{\frac{1}{n} - \overline{p}^{(K)}_i}\\
	&+ \frac{1}{\mu} \sum_{i \in \mathcal{I}^C} (f_i(\xstar) - f_i(\xstar_i)) \abs{\frac{1}{n} -\overline{p}^{(K)}_i }.
\end{split}
\end{align}
Using the smoothness of the global objective, we can bound the terms in the first summation in~\eqref{eq:bound_dist_minimizer_eq4},
\begin{align}
\label{eq:bound_dist_minimizer_eq5}
\begin{split}
	\frac{1}{\mu} (f_i(\xstar_K) - f_i(\xstar_i)) \abs{\frac{1}{n} - \overline{p}^{(K)}_i} \leq \frac{\kappa}{2} \norm{ \xstar_K - \xstar_i }^2 \abs{\frac{1}{n} - \overline{p}^{(K)}_i},
\end{split}
\end{align}
and similarly for the terms in the second summation in~\eqref{eq:bound_dist_minimizer_eq4},
\begin{align}
\label{eq:bound_dist_minimizer_eq6}
\begin{split}
	\frac{1}{\mu} (f_i(\xstar) - f_i(\xstar_i)) \abs{\frac{1}{n} - \overline{p}^{(K)}_i} \leq \frac{\kappa}{2} \norm{ \xstar - \xstar_i }^2 \abs{\frac{1}{n} - \overline{p}^{(K)}_i}.
\end{split}
\end{align}
Substituting~\eqref{eq:bound_dist_minimizer_eq5} and~\eqref{eq:bound_dist_minimizer_eq6} back into~\eqref{eq:bound_dist_minimizer_eq4}, we have
\begin{align}
\label{eq:bound_dist_minimizer_eq7}
\begin{split}
	\norm{ \xstar_K - \xstar }^2 \leq& \frac{\kappa}{2} \sum_{i \in \mathcal{I}} \norm{ \xstar_K - \xstar_i }^2 \abs{\frac{1}{n} - \overline{p}^{(K)}_i}\\
	&+ \frac{\kappa}{2} \sum_{i \in \mathcal{I}^C} \norm{ \xstar - \xstar_i }^2 \abs{\frac{1}{n} -\overline{p}^{(K)}_i }.
\end{split}
\end{align}
Note that there exists an index $j \in [n]$ such that $\norm{ \xstar_K - \xstar_j } \leq \norm{ \xstar- \xstar_j}$. To see this, suppose for the sake of a contradiction that $\norm{ \xstar_K - \xstar_j} > \norm{ \xstar - \xstar_j }$ for all $j \in [n]$. 
Since the local objectives are strongly convex, this implies that there exists a point $\xstar$ such that $f_j(\xstar) < f_j(\xstar_K)$ for all $j \in [n]$.
Therefore, $F_K(\xstar) < F_K(\xstar_K)$, which contradicts the definition of $\xstar_K$. Hence there exists $j \in [n]$ such that
\begin{equation}
\label{eq:bound_dist_minimizer_eq8}
	\norm{ \xstar_K - \xstar_j } \leq \norm{ \xstar- \xstar_j}.
\end{equation}
Using the triangle inequality and~\eqref{eq:bound_dist_minimizer_eq8}
\[
	\norm{ \xstar_K - \xstar_i} \leq \overline{S}.
\]
Similarly, using the triangle inequality
\[
	\norm{ \xstar_i - \xstar} \leq  \overline{S}.
\]
Therefore, we can simplify~\eqref{eq:bound_dist_minimizer_eq7} as
\begin{align}
\label{eq:bound_dist_minimizer_eq9}
\begin{split}
	\norm{ \xstar_K - \xstar }^2 \leq& \frac{\overline{S}^2 \kappa}{2} \sum^n_{i =1}\abs{\frac{1}{n} - \overline{p}^{(K)}_i}.
\end{split}
\end{align}
Taking the square-root on each side of~\eqref{eq:bound_dist_minimizer_eq9} gives the desired result.
\end{proof}

\subsection{Preliminaries}
Before proceeding to the proofs of Theorems~\ref{th:const_convergence_opt} and~\ref{th:convergence_opt}, we derive some preliminary results here. Then we give the proof of Theorem~\ref{th:const_convergence_opt} followed by the proof of Theorem~\ref{th:convergence_opt} in the remainder of this section.

\begin{lemma}
\label{lem:bounded_suboptimality}
Suppose Assumptions~\ref{ass:base_obj} and \ref{ass:step_size_bound} are satisfied.
Then for all $k \geq 0$,
\[
	\norm{ \xbar[k] - \xstar_K } \leq \frac{L}{\mu},
\]
where $L$ is defined in Theorem~\ref{th:bounded_iterates_gradients}, and $\xstar_K$ is the minimizer of the re-weighted objective defined in Definition~\ref{def:reweighted_obj}.
\end{lemma}
\begin{proof}
Using the strong convexity of the global objective and the fact that $\xstar_K$ is the minimizer of the re-weighted objective $ \sum^n_{i=1} \overline{p}^{(K)}_i f_i( \cdot)$, we have that
\[
	\norm{ \xbar[k] - \xstar_K } \leq \frac{1}{\mu} \norm{ \sum^n_{i=1} \overline{p}^{(K)}_i \nabla f_i ( \xbar[k])}.
\]
Using the convexity of the norm and substituting the gradient upper bound from Theorem~\ref{th:bounded_iterates_gradients} gives the desired result.
\end{proof}

\begin{lemma}
\label{lem:inner_prod}
Suppose Assumptions~\ref{ass:base_obj} and~\ref{ass:step_size_bound} are satisfied. Define
\begin{align*}
	\gamma_i[k] \defeq&	\kappa L C \norm{x_i[0]}_1 q^k \\
	\chi_i[k] \defeq& \kappa L^2 C \sum^k_{s=0} q^{k - s} \alpha_i[s] \delta_i[s]
\end{align*}
where $q \in (0, 1)$ and $C > 0$ are defined in Theorem~\ref{th:avg_rate}.
Then for all $i=1, \ldots, n$ it holds that
\begin{align*}
	\langle \nabla f_i(z_i[k]), \xbar[k] - \xstar_K \rangle \geq&\ \mu \norm{\xbar[k] - \xstar_K}^2 \\
	& - \gamma_i[k] - \chi_i[k] \\
	& + \langle \nabla f_i(\xstar_K), \xbar[k] - \xstar_K \rangle.
\end{align*}
\end{lemma}
\begin{proof}
Begin by re-writing the inner product
\begin{align}
\label{eq:lem_inner_prod_pf1}
\begin{split}
	\langle \nabla f_i(z_i[k])&, \xbar[k] - \xstar_K \rangle \\
	=& \langle \nabla f_i(z_i[k]) - \nabla f_i(\xbar[k]), \xbar[k] - \xstar_K \rangle \\
	& + \langle \nabla f_i(\xbar[k]), \xbar[k] - \xstar_K \rangle.
\end{split}
\end{align}
Using the Lipschitz-smoothness of the objectives, we have
\begin{align}
\label{eq:lem_inner_prod_pf2}
\begin{split}
	\langle \nabla f_i(z_i[k])& - \nabla f_i(\xbar[k]), \xbar[k] - \xstar_K \rangle \\
	\geq& -M \norm{z_i[k] - \xbar[k]} \norm{\xbar[k] - \xstar_K }.
\end{split}
\end{align}
Making use of Lemma~\ref{lem:bounded_suboptimality}, we can simplify~\eqref{eq:lem_inner_prod_pf2} as
\begin{align}
\label{eq:lem_inner_prod_pf3}
\begin{split}
	\langle \nabla f_i(z_i[k])& - \nabla f_i(\xbar[k]), \xbar[k] - \xstar_K \rangle \\
	\geq& - \kappa L \norm{z_i[k] - \xbar[k]}.
\end{split}
\end{align}
Applying the result of Theorem~\ref{th:avg_rate} in~\eqref{eq:lem_inner_prod_pf3}, and substituting the gradient bounds from Theorem~\ref{th:bounded_iterates_gradients}, we have
\begin{align*}
\begin{split}
	\langle \nabla& f_i(z_i[k]) - \nabla f_i(\xbar[k]), \xbar[k] - \xstar_K \rangle \\
	& \geq - \left( \kappa L C \right) \left( \norm{ x_i[0]}_{1} q^{k} +  L \sum^{k}_{s=0} q^{k - s} \alpha_i[s] \delta_i[s] \right),
\end{split}
\end{align*}
thereby bounding the first term in~\eqref{eq:lem_inner_prod_pf1}.
Using the strong-convexity of the objectives, we can bound the second term in~\eqref{eq:lem_inner_prod_pf1} as
\begin{align}
\label{eq:lem_inner_prod_pf4}
\begin{split}
	\langle \nabla f_i(\xbar[k]), \xbar[k] - \xstar_K \rangle \geq&\ \langle \nabla f_i(\xstar_K), \xbar[k] - \xstar_K \rangle \\
	&\ + \mu \norm{ \xbar[k] - \xstar_K }^2.
\end{split}
\end{align}
\end{proof}

\begin{lemma}
\label{lem:contraction}
Suppose Assumptions~\ref{ass:base_obj} and~\ref{ass:step_size_bound} are satisfied.
For any integer $K > 0$, it holds that
\begin{align*}
	\frac{1}{n K} \sum^{K - 1}_{k=0} & \sum^n_{i=1} \alpha_i[k] \delta_i[k] \langle \nabla f_i(\xstar_K), \xbar[k] - \xstar_K \rangle \geq 0,
\end{align*}
where $\xstar_K$ is the minimizer of the re-weighted objective defined in Definition~\ref{def:reweighted_obj}.
\end{lemma}
\begin{proof}
Begin by re-writing the inner product
\begin{align}
\label{eq:lem_contraction_pf1}
\begin{split}
	\langle \nabla f_i(\xstar_K), \xbar[k] - \xstar_K \rangle =& \langle \nabla f_i(\xstar_K), \xbar[K] - \xstar_K \rangle \\
	& + \langle \nabla f_i(\xstar_K), \xbar[k] - \xbar[K] \rangle.
\end{split}
\end{align}
From Lemma~\ref{lem:bounded_suboptimality}, we have
\begin{align}
\label{eq:lem_contraction_pf2}
\begin{split}
	\frac{1}{n K} \sum^{K 	- 1}_{k=0} & \sum^n_{i=1} \alpha_i[k] \delta_i[k] \langle \nabla f_i(\xstar_K), \xbar[K] - \xstar_K \rangle \\
	&\geq - \frac{L}{\mu} \norm{ \frac{1}{nK} \sum^n_{i=1} \nabla f_i(\xstar_K) \left( \sum^{K- 1}_{k=0} \alpha_i[k] \delta_i[k] \right)}.
\end{split}
\end{align}
Recalling that $p^{(K)}_i \defeq \sum^{K- 1}_{k=0} \alpha_i[k] \delta_i[k]$, and that $\xstar_K$ is the minimizer of the re-weighted objective $\sum^n_{i=1} f_i( \cdot) p^{(K)}_i$, it follows that the right-hand-side of~\eqref{eq:lem_contraction_pf2} vanishes, and
\begin{align}
\label{eq:lem_contraction_pf3}
\begin{split}
	\frac{1}{n K} \sum^{K- 1}_{k=0} & \sum^n_{i=1} \alpha_i[k] \delta_i[k] \langle \nabla f_i(\xstar_K), \xbar[K] - \xstar_K \rangle \geq 0.
\end{split}
\end{align}
Now turning our attention to the second term on the right-hand side of~\eqref{eq:lem_contraction_pf1}, we have
\begin{align*}
\begin{split}
	\langle \nabla f_i(\xstar_K)&, \xbar[k] - \xbar[K] \rangle \\
	& = \left\langle \nabla f_i(\xstar_K), \sum^{K- 1}_{\ell = k} \frac{1}{n} \sum^n_{i=1} \alpha_i[\ell] \delta_i[\ell] \nabla f_i (z_i[\ell]) \right\rangle.
\end{split}
\end{align*}
Define the positive integer $k^\prime$ as
\[
	k^\prime  \defeq \argmin_{k \in \{ 0, 1, \ldots, K - 1\}} \left\langle \nabla f_i(\xstar_K), \sum^{K- 1}_{\ell = k} \frac{1}{n} \sum^n_{i=1} \alpha_i[\ell] \delta_i[\ell] \nabla f_i (z_i[\ell]) \right\rangle,
\]
and the corresponding vector, $v_K \in \R^d$,
\[
	v_K \defeq \sum^{K - 1}_{\ell = k^\prime} \frac{1}{n} \sum^n_{i=1} \alpha_i[\ell] \delta_i[\ell] \nabla f_i (z_i[\ell]).
\]
It holds for all $k = 0, 1, \dots, K- 1$ that
\begin{align*}
\begin{split}
	\langle \nabla f_i(\xstar_K)&, \xbar[k] - \xbar[K] \rangle \geq \langle \nabla f_i(\xstar_K), v_K \rangle.
\end{split}
\end{align*}
Therefore,
\begin{align}
\label{eq:lem_contraction_pf4}
\begin{split}
	\frac{1}{n K} \sum^{K- 1}_{k=0} & \sum^n_{i=1} \alpha_i[k] \delta_i[k] \langle \nabla f_i(\xstar_K), \xbar[k] - \xstar[K] \rangle \\
	&\geq - \frac{\norm{v_K}}{K} \norm{ \frac{1}{n} \sum^n_{i=1} \nabla f_i(\xstar_K) \left( \sum^{K- 1}_{k=0} \alpha_i[k] \delta_i[k] \right)}.
\end{split}
\end{align}
Note that, from Theorem~\ref{th:bounded_iterates_gradients}, we have
\begin{align}
\label{eq:lem_contraction_pf5}
	\norm{v_K} \leq K L \frac{1}{n} \sum^n_{i=1} \alpha_i[0] .
\end{align}
Substituting~\eqref{eq:lem_contraction_pf5} into~\eqref{eq:lem_contraction_pf4}, gives
\begin{align}
\label{eq:lem_contraction_pf6}
\begin{split}
	\frac{1}{n K} &\sum^{K- 1}_{k=0} \sum^n_{i=1} \alpha_i[k] \delta_i[k] \langle \nabla f_i(\xstar_K), \xbar[k] - \xstar[K] \rangle \\
	&\geq - \norm{ \frac{1}{n} \sum^n_{i=1} \nabla f_i(\xstar_K) \left( \sum^{K- 1}_{k=0} \alpha_i[k] \delta_i[k] \right)} \frac{L}{n} \sum^n_{i=1} \alpha_i[0] .
\end{split}
\end{align}
Recalling that $p^{(K)}_i \defeq \sum^{K- 1}_{k=0} \alpha_i[k] \delta_i[k]$, and that $\xstar_K$ is the minimizer of the re-weighted objective $\sum^n_{i=1} f_i( \cdot) p^{(K)}_i$, it follows that the right-hand side of~\eqref{eq:lem_contraction_pf6} vanishes, and
\begin{align}
\label{eq:lem_contraction_pf7}
\begin{split}
	\frac{1}{n K} \sum^{K- 1}_{k=0} & \sum^n_{i=1} \alpha_i[k] \delta_i[k] \langle \nabla f_i(\xstar_K), \xbar[k] - \xbar[K] \rangle \geq 0.
\end{split}
\end{align}
Summing~\eqref{eq:lem_contraction_pf7} and~\eqref{eq:lem_contraction_pf3} together gives the desired result.
\end{proof}

\begin{lemma}
\label{lem:bounded_summations_const}
Suppose Assumptions~\ref{ass:base_obj},~\ref{ass:step_size_bound}, and~\ref{ass:step_size_const} are satisfied.
Define
\begin{align*}
	b_1[K] &\defeq L \sum^{K - 1}_{k=0} \left( \frac{1}{n} \sum^n_{i=1} \alpha_i[k] \delta_i[k] \right)^2 \\
	b_2[K] &\defeq 2L\sum^{K - 1}_{k=0} \left( \frac{1}{n} \sum^n_{i=1} \alpha_i[k] \delta_i[k] \gamma_i[k] \right) \\
	b_3[K] &\defeq 2L\sum^{K - 1}_{k=0} \left( \frac{1}{n} \sum^n_{i=1} \alpha_i[k] \delta_i[k] \chi_i[k] \right),
\end{align*}
where $\gamma_i[k]$ and $\chi_i[k]$ are given in Lemma~\ref{lem:inner_prod}.
There exist finite constants $A_1, A_2, A_3 > 0$, such that,
\[
	b_1[K] \leq \frac{A_1}{K^{2 \theta - 1}}, \quad b_2[K] \leq \frac{A_2}{K^\theta}, \quad b_3[K] \leq \frac{A_3}{K^{2\theta - 1}}.
\]
\end{lemma}
\begin{proof}
From Assumption~\ref{ass:step_size_const}, we have
\begin{align*}
	b_1[K] \leq L \left( \frac{B}{n} \sum^n_{i=1} w_i \right)^2 \frac{1}{K^{2 \theta -1 }}.
\end{align*}
Letting $A_1 \defeq \left( \frac{ \sqrt{L} B}{n} \sum^n_{i=1}w_i \right)^2$, we have $b_1[K] \leq \frac{A_1}{K^{2 \theta - 1}}$.
Now to bound $b_2[K]$, note that, given Assumption~\ref{ass:step_size_const}, we have
\[
	\sum^{K - 1}_{k=0} \left( \alpha_i[k] \delta_i[k] \right) q^k \leq \frac{\alpha_i}{1 - q}.
\]
Letting $A_2 \defeq \frac{2 \kappa L^2 C \norm{x_i[0]}  \left( \frac{B}{n} \sum^n_{i=1} w_i \right)}{(1-q)}$, we have $b_2[K] \leq \frac{A_2}{K^\theta}$.
Lastly, to bound $b_3[K]$,  it follows from Assumption~\ref{ass:step_size_const}, that
\begin{align*}
	\sum^{K-1}_{k=0} \chi_i[k] \left( \alpha_i[k] \delta_i[k] \right) \leq \alpha_i^2 \kappa L^2 C \sum^{K - 1}_{k=0} \sum^{k}_{s=0} q^{k - s} \leq \frac{\alpha_i^2 \kappa L^2 C K}{1- q}.
\end{align*}
Letting $A_3 \defeq \frac{2 \kappa L^3 C \left( \frac{B}{n} \sum^n_{i=1} w_i \right)^2 }{(1-q)}$, we have $b_3[K] \leq \frac{A_3}{K^{2 \theta - 1}}$.
\end{proof}

\begin{lemma}
\label{lem:bounded_summations}
Suppose Assumptions~\ref{ass:base_obj},~\ref{ass:step_size_bound}, and~\ref{ass:step_size_decay} are satisfied.
Define
\begin{align*}
	b_1[K] &\defeq L \sum^{K - 1}_{k=0} \left( \frac{1}{n} \sum^n_{i=1} \alpha_i[k] \delta_i[k] \right)^2 \\
	b_2[K] &\defeq 2L\sum^{K - 1}_{k=0} \left( \frac{1}{n} \sum^n_{i=1} \alpha_i[k] \delta_i[k] \gamma_i[k] \right) \\
	b_3[K] &\defeq 2L\sum^{K - 1}_{k=0} \left( \frac{1}{n} \sum^n_{i=1} \alpha_i[k] \delta_i[k] \chi_i[k] \right),
\end{align*}
where $\gamma_i[k]$ and $\chi_i[k]$ are given in Lemma~\ref{lem:inner_prod}.
There exists a finite constant $A > 0$, such that for all $K \geq 0$,
\[
	b_1[K] + b_2[K] + b_3[K] \leq A.
\]
\end{lemma}
\begin{proof}
First note that the sequences $b_1[K]$, $b_2[K]$, and $b_3[K]$ are all monotonically increasing with $K$.
Therefore, if we can show that the sequences are bounded, then it follows that they are also convergent, and their respective limits serve as upper bounds.
From Assumption~\ref{ass:step_size_decay} and Remark~\ref{rem:step_size_decay}, it immediately follows that the sequence $b_1[K]$ is bounded, and therefore convergent.
Let $A'_1 \defeq \lim_{K \to \infty} b_1[K]$. Consequently, $b_1[K] \leq A'_1$ for all $K \geq 0$.
Now to bound $b_2[K]$, note that, given Assumption~\ref{ass:step_size_decay}, it holds that
\[
	\sum^{\infty}_{k=0} \left( \alpha_i[k] \delta_i[k] \right) q^k \leq \frac{\alpha_i[0]}{1 - q} < \infty.
\]
Let $A'_2 \defeq \frac{2 \kappa L^2 C \norm{ x_i[0]} }{1-q} \frac{1}{n} \sum^n_{i=1} \alpha_i[0]$. It follows that $b_2[K] \leq A'_2$ for all $K \geq 0$.
Lastly,  to bound $b_3[K]$, it follows from~\cite[Lemma 3.1]{ram2010distributed} and Assumption~\ref{ass:step_size_decay}, that
\begin{align*}
	\sum^\infty_{k=0} \chi_i[k] \left( \alpha_i[k] \delta_i[k] \right) \leq \kappa L^2 C \sum^\infty_{k=0} \sum^{k}_{s=0} q^{k - s} \left( \alpha_i[s] \delta_i[s] \right)^2 < \infty.
\end{align*}
Therefore, $b_3[K]$ is bounded and convergent. Let $A'_3 \defeq \lim_{K \to \infty} b_3[K]$. Then $b_3[K] \leq A'_3 < \infty$ for all $K \geq 0$.
Defining $A \defeq A'_1 + A'_2 + A'_3$ gives the desired result.
\end{proof}

\subsection{Proof of Theorem~\ref{th:const_convergence_opt}}
\begin{proof}[\unskip\nopunct]
Recall the update equation~\eqref{eq:asynch_itr_1} given by
\begin{align*}
    \bm{x}[k + 1] &= \Pbar[k] \left( \bm{x}[k] - \Grad{k} \right).
\end{align*}
Since the matrices $\Pbar[k]$ are column stochastic, we can multiply each side of~\eqref{eq:asynch_itr_1} by $\ones^T / n$ to get
\begin{align}
\label{eq:avg_asynch_itr_1}
\begin{split}
    \xbar[k + 1] =&\ \xbar[k] - \sum^{n}_{i = 1} \frac{ \alpha_i[k] \delta_i[k]}{n} \nabla f_i(z_i[k]).
\end{split}
\end{align}
Subtracting $\xstar_K$ from each side of~\eqref{eq:avg_asynch_itr_1} and taking the squared norm
\begin{align}
\label{eq:avg_asynch_itr_1}
\begin{split}
	\norm{ \xbar[k + 1] - \xstar_K }^2& \leq \norm{ \xbar[k] - \xstar_K }^2 \\
	& - \frac{2}{n} \sum^n_{i=1} \alpha_i[k] \delta_i[k] \langle \nabla f_i(z_i[k]), \xbar[k] - \xstar_K \rangle \\
	& + \norm{ \frac{1}{n} \sum^n_{i=1} \alpha_i[k] \delta_i[k] \nabla f_i(z_i[k]) }^2.
\end{split}
\end{align}
Note that, from Theorem~\ref{th:bounded_iterates_gradients}, we have
\[
	\norm{ \frac{1}{n} \sum^n_{i=1} \alpha_i[k] \delta_i[k] \nabla f_i(z_i[k]) }^2 \leq \left( \frac{L}{n} \sum^n_{i=1} \alpha_i[k] \delta_i[k] \right)^2,
\]
thereby bounding the last term in~\eqref{eq:avg_asynch_itr_1}.
Additionally, making use of Lemma~\ref{lem:inner_prod}, it follows that
\begin{align}
\label{eq:avg_asynch_itr_2}
\begin{split}
	\norm{ \xbar[k + 1] - \xstar_K }^2 \leq& \norm{ \xbar[k] - \xstar_K }^2 + \left( \frac{L}{n} \sum^n_{i=1} \alpha_i[k] \delta_i[k] \right)^2 \\
	& - 2 \mu \norm{\xbar[k] - \xstar_K}^2 \left( \frac{1}{n} \sum^n_{i=1} \alpha_i[k] \delta_i[k] \right) \\
	& - \frac{2}{n} \sum^n_{i=1} \alpha_i[k] \delta_i[k] \langle \nabla f_i(\xstar_K), \xbar[k] - \xstar_K \rangle \\
	& + \frac{2}{n} \sum^n_{i=1} \alpha_i[k] \delta_i[k] (\gamma_i[k] + \chi_i[k]).
\end{split}
\end{align}
Rearranging terms, averaging each side of~\eqref{eq:avg_asynch_itr_2} across time indices, and making use of Lemma~\ref{lem:contraction} gives
\begin{align}
\label{eq:avg_asynch_itr_3}
\begin{split}
	\frac{2 \mu}{K} &\sum^{K -1}_{k=0} \norm{\xbar[k] - \xstar_K}^2 \left( \frac{1}{n} \sum^n_{i=1} \alpha_i[k] \delta_i[k] \right) \\
	\leq& \frac{1}{K} \sum^{K -1}_{k=0} \left( \norm{ \xbar[k] - \xstar_K }^2 - \norm{\xbar[k + 1] - \xstar_K}^2 \right) \\
	& + \frac{1}{K} \sum^{K - 1}_{k=0} \left( \frac{2}{n} \sum^n_{i=1} \alpha_i[k] \delta_i[k] \left( \gamma_i[k] + \chi_i[k] \right) \right) \\
	& + \frac{1}{K} \sum^{K - 1}_{k=0} \left( \frac{L}{n} \sum^n_{i=1} \alpha_i[k] \delta_i[k] \right)^2.
\end{split}
\end{align}
Noticing that we have a telescoping sum on the right hand side of~\eqref{eq:avg_asynch_itr_3}, and making use of Lemma~\ref{lem:bounded_summations_const} and Assumption~\ref{ass:step_size_const}, it follows that
\begin{align*}
	\frac{1}{K} \sum^{K -1}_{k=0} \norm{\xbar[k] - \xstar_K}^2 \leq& \frac{1}{K^{1 - \theta}} \left(\frac{n \left( \norm{ \xbar[0] - \xstar_K }^2 \right)}{2 \mu B } \right) \\
	& + \frac{1}{ K^\theta} \left( \frac{n (A_1 + A_3)}{2 \mu B} \right) +  \frac{1}{ K} \left( \frac{n A_2}{2 \mu B} \right)
\end{align*}
where $\theta \in (0, 1)$ is defined in Assumption~\ref{ass:step_size_const}.
\end{proof}

\subsection{Proof of Corollary~\ref{cor:const_semi_synch_convergence_opt}}
\begin{proof}[\unskip\nopunct]
If $\tauratemax = 1$, then each agent performs a gradient update in each iteration. In particular, $\delta_i[k] = 1$ for all $k \geq 0$ and $i=1,\ldots,n$.
Using the fact that $w_i = 1$ for all $i=1,\ldots,n$ (agents use the same factor in their local step-sizes), it follows that $p^{(K)}_i = p^{(K)}_j$ for all $i,j=1,\ldots,n$.
Hence, the minimizer of the re-weighted objective reduces to that of the original (unbiased) objective, \ie, $\xstar_K = \xstar$.
Substituting into the result of Theorem~\ref{th:const_convergence_opt} gives the desired result.
\end{proof}

\subsection{Proof of Corollary~\ref{cor:const_congergence_opt_known_udpate_rates}}
\begin{proof}[\unskip\nopunct]
Note that
\[
	p^{(K)}_i \defeq \sum^{K - 1}_{k=0} \alpha_i[k] \delta_i[k] = \frac{w_i B}{K^\theta} c_i[K - 1].
\]
Given the choice of $w_i$, it follows that
\[
	p^{(K)}_i = \frac{B}{K^{\theta - 1}},
\]
and is agnostic of the index $i$.
Therefore, $p^{(K)}_i = p^{(K)}_j$ for all $i,j=1,\ldots,n$.
Hence, the minimizer of the re-weighted objective reduces to that of the original (unbiased) objective, \ie, $\xstar_K = \xstar$.
Substituting into the result of Theorem~\ref{th:const_convergence_opt} gives the desired result.
\end{proof}

\subsection{Proof of Theorem~\ref{th:convergence_opt}}

\begin{proof}[\unskip\nopunct]
The proof of Theorem~\ref{th:convergence_opt} is identical to that of Theorem~\ref{th:const_convergence_opt} up to~\eqref{eq:avg_asynch_itr_3}.
Noticing that we have a telescoping sum on the right hand side of~\eqref{eq:avg_asynch_itr_3}, and making use of Lemma~\ref{lem:bounded_summations} and Remark~\ref{rem:step_size_decay}, it follows that
\begin{align*}
	\frac{1}{K} & \sum^{K -1}_{k=0} \norm{\xbar[k] - \xstar_K}^2 \leq& \frac{1}{K^{1 - \theta}} \left(\frac{n \left( \norm{ \xbar[0] - \xstar_K }^2 + A \right)}{2 \mu B } \right),
\end{align*}
where $\theta \in (0.5, 1)$ is defined in Assumption~\ref{ass:step_size_decay}.
\end{proof}

\subsection{Proof of Corollary~\ref{cor:semi_synch_convergence_opt}}
\begin{proof}[\unskip\nopunct]
If $\tauratemax = 1$, then each agent performs a gradient update in each iteration. In particular, $\delta_i[k] = 1$ for all $k \geq 0$ and $i=1,\ldots,n$.
Using the fact that $w_i = 1$ for all $i=1,\ldots,n$ (agents use the same factor in their local step-sizes), it follows that $p^{(K)}_i = p^{(K)}_j$ for all $i,j=1,\ldots,n$.
Hence, the minimizer of the re-weighted objective reduces to that of the original (unbiased) objective, \ie, $\xstar_K = \xstar$.
Substituting into the result of Theorem~\ref{th:convergence_opt} gives the desired result.
\end{proof}

\subsection{Proof of Corollary~\ref{cor:congergence_opt_known_udpate_rates}}
\begin{proof}[\unskip\nopunct]
Note that
\[
	p^{(K)}_i \defeq \sum^{K - 1}_{k=0} \alpha_i[k] \delta_i[k] = \sum^{c_i[K - 1] - 1}_{k=0} \frac{w_i B}{(k + 1)^\theta}.
\]
Given the choice of $w_i$, it follows that
\[
	p^{(K)}_i = \sum^{K - 1}_{k=0} \frac{B}{(k + 1)^\theta},
\]
and is agnostic of the index $i$.
Therefore, $p^{(K)}_i = p^{(K)}_j$ for all $i,j=1,\ldots,n$.
Hence, the minimizer of the re-weighted objective reduces to that of the original (unbiased) objective, \ie, $\xstar_K = \xstar$.
Substituting into the result of Theorem~\ref{th:convergence_opt} gives the desired result.
\end{proof}

\section{Experiments}
\label{sec:experiments}

Next, we report experiments on a high performance computing cluster. In these experiments, each agent is implemented as a process running on a dedicated CPU core, and each agent runs on a different server. Communication between servers happens over an InfiniBand network.
The code to reproduce these experiments is available online;\footnote{https://github.com/MidoAssran/maopy} all code is written in Python, and the Open-MPI distribution is used with Python bindings (mpi4py) for message passing.

We report two sets of experiments. The first set involves solving a least-squares regression problem using synthetic data. The aim of these experiments is to validate the theory developed in the sections above for AGP. The second set of experiments involves solving a regularized multinomial logistic regression problem on a real dataset. In these experiments we compare AGP with three synchronous methods: Push DIGing (PD)~\cite{nedic2017achieving}, Extra Push (EP)~\cite{zeng2015extrapush}, and Synchronous (Sub)Gradient-Push (SGP)~\cite{nedich2015distributed}. Both PD and EP use gradient tracking to achieve stronger theoretical convergence guarantees at the cost of additional communication overhead. We also compare with Asy-SONATA~\cite{tian2018achieving}, an asynchronous method that incorporates gradient tracking and which appeared online during the review process of this paper. Note that all methods that use gradient tracking (PD, EP, and Asy-SONATA) require additional memory at each agent and also have a communication overhead per-iteration which is twice that of SGP and AGP.

\subsection{Synthetic Dataset}
To validate some of the theory developed in previous sections, we first report experiments on a linear least-squares regression problem using synthetic data.
The objective is to minimize, over parameters $\bm{w}$, the function:
\begin{equation}
	F(\bm{w}) \defeq \frac{1}{D} \sum^D_{\ell=1}(w^T_j x^\ell - y^\ell)^2, \label{eqn:linear-regression}
\end{equation}
where $D = \num[group-separator={,}]{2560000}$ is the number of training instances in the dataset,  $x^l \in \R^{50}$ and $y^l \in \R^{1}$ correspond to the $l^{th}$ training instance feature and label vectors respectively, and $\bm{w} \in \R^{50}$ are the model parameters.
We generate the data $\{(x^\ell, y^\ell)\}_{\ell=1}^D$ using the technique suggested in~\cite{lenard1984randomly}.

\begin{figure}[!t]
\centering
\includegraphics[width=.6\columnwidth]{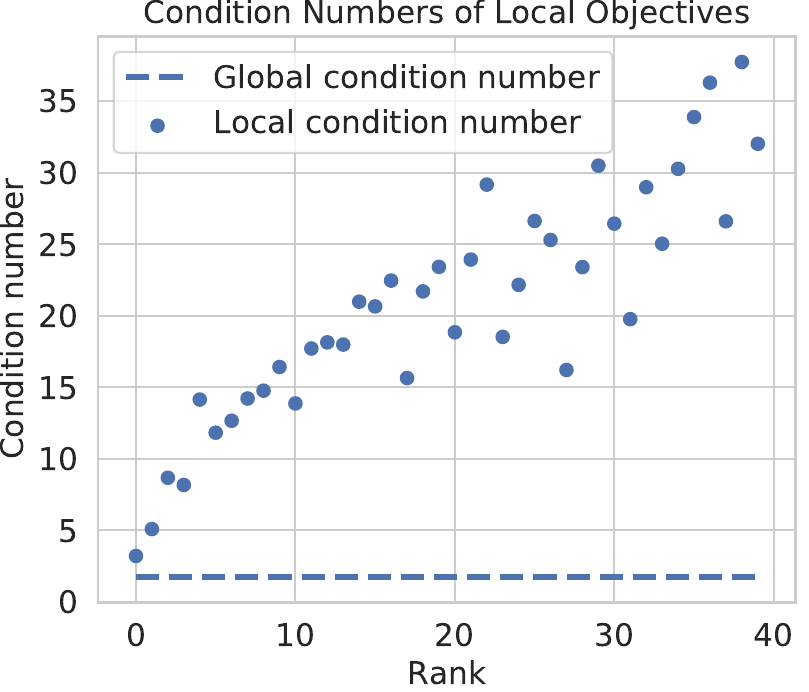}
\caption{Condition numbers of local objective functions for a 40-agent partition of the synthetic dataset. The dashed line shows the condition number of the global objective.}
\label{fig:condition_numbers}
\end{figure}

The $D$ data samples are partitioned among the $n$ agents. The local objective function $f_i$ at agent $v_i$ is similar to that in \eqref{eqn:linear-regression} but the sum over $l$ only involves those training instances assigned to agent $v_i$.
The condition number of the global objective is approximately $2$.
The condition number of individual agents' local objectives is diverse and depends on the data-partition.
Figure~\ref{fig:condition_numbers} shows the local objective conditioning for a $40$-agent partition of the dataset.
The condition numbers of the local objectives are approximately uniformly spaced in the interval $(3, 37)$.

\begin{figure}[!t]
\centering
\includegraphics[width=\columnwidth]{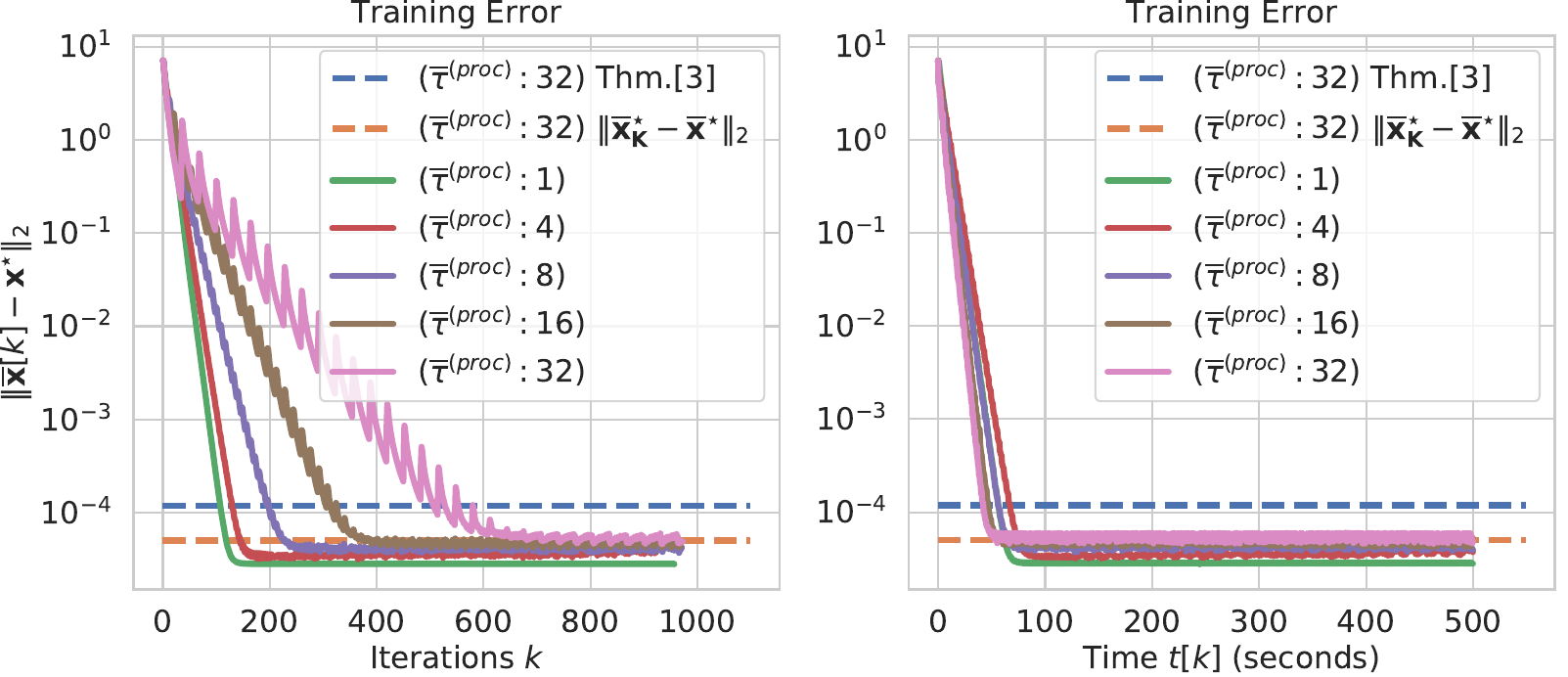}
\caption{Convergence of Asynchronous Gradient Push for a 40-agent ring-network with various degrees of asynchrony (quantified by $\tauratemax$).
The dashed blue bar corresponds to the $\norm{ \xstar_K - \xstar}$ bound from Theorem~\ref{th:bound_dist_minimziers}, where the reweighing values $\{\overline{p}^{(K)}_i\}$ are computed from the experiment corresponding to $\tauratemax=32$.
The dashed orange bar corresponds to the true value of $\norm{ \xstar_K - \xstar}$ for the experiments corresponding to $\tauratemax=32$.}
\label{fig:agp_train_tau}
\end{figure}
\begin{figure}[!t]
\centering
\includegraphics[width=.6\columnwidth]{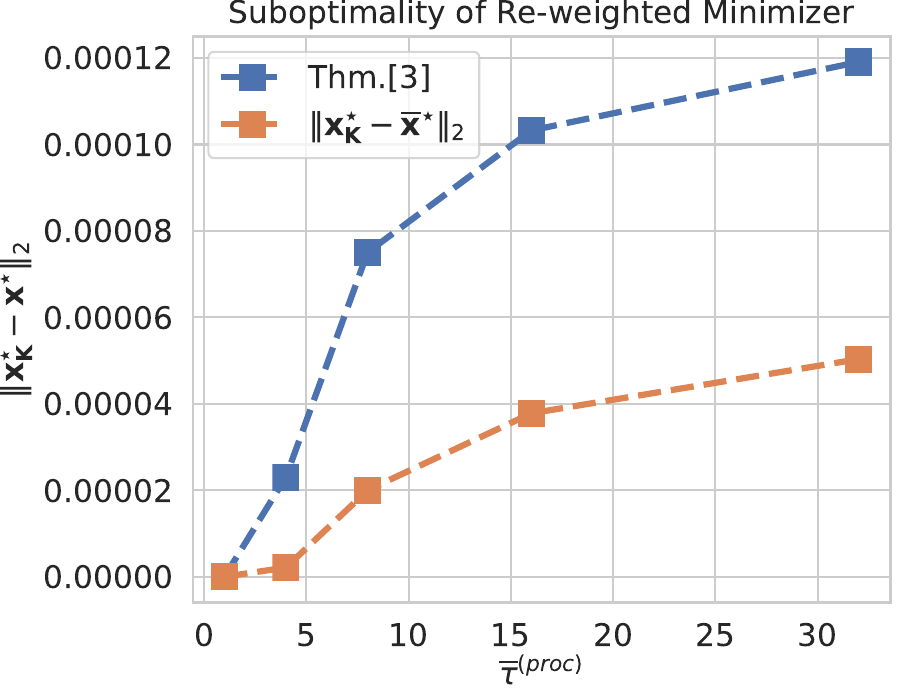}
\caption{Distance between the minimizer of the re-weighted objective $\xstar_K$ and the original (unbiased) objective for different choices of $\tauratemax$.
The blue points depict the bound in Theorem~\ref{th:bound_dist_minimziers}, and the red points depict the true quantity.}
\label{fig:x_K_sweep_tau}
\end{figure}

During training, agent $v_i$ logs the values of $z_i$ and the time after every update. Post training, we analytically compute the minimizer of the re-weighted objective defined in Definition~\ref{def:reweighted_obj}. 
To validate the bound on the distance between the minimizer of the re-weighted objective and the original unbiased objective (cf.~Theorem~\ref{th:bound_dist_minimziers}), we run AGP for different choices of $\tauratemax$.
We control $\tauratemax$ by forcing an agent to block if it completes $\tauratemax$ iterations while another agent still hasn't completed a single iteration in the same time interval; thus, in the worst case scenario, a fast agent can complete $\tauratemax$ iterations for every iteration completed by a slow agent.\footnote{For the purpose of this experiment, we artificially delay half of the agents in the network by $500$ ms each iteration, and implement $\tauratemax$ programmatically using non-blocking barrier operations (which are a part of the MPI-3 standard).
In particular, each agent tests a non-blocking barrier request at each local iteration.
If the test is passed, then a new non-blocking barrier request object is created.
If the test is not passed and more than $\tauratemax$ local iterations have gone by since the last test was passed, then the agent blocks and waits for the barrier-test to pass.
In this way, no more than $\tauratemax$ iterations can be performed by the network in the time it takes any single agent to complete one local iteration.}
In Fig.~\ref{fig:agp_train_tau} we show the convergence of AGP for different values of $\tauratemax$.
We use a directed ring network in this example to examine the worst-case scenario.

\begin{figure*}[!t]
\centering
\subfloat[\label{fig:0ms-world-sweep} No artificial delay]{\includegraphics[width=0.24\textwidth]{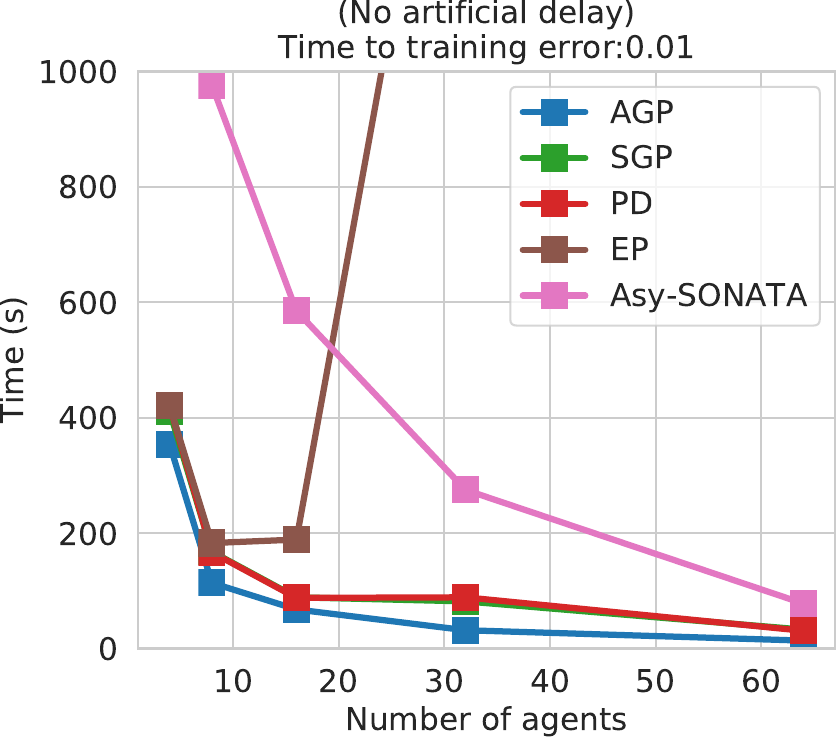}} \hspace*{0.005\textwidth}
\subfloat[\label{fig:125ms-world-sweep} $125$ms delay injected at agent $v_{2}$]{\includegraphics[width=0.24\textwidth]{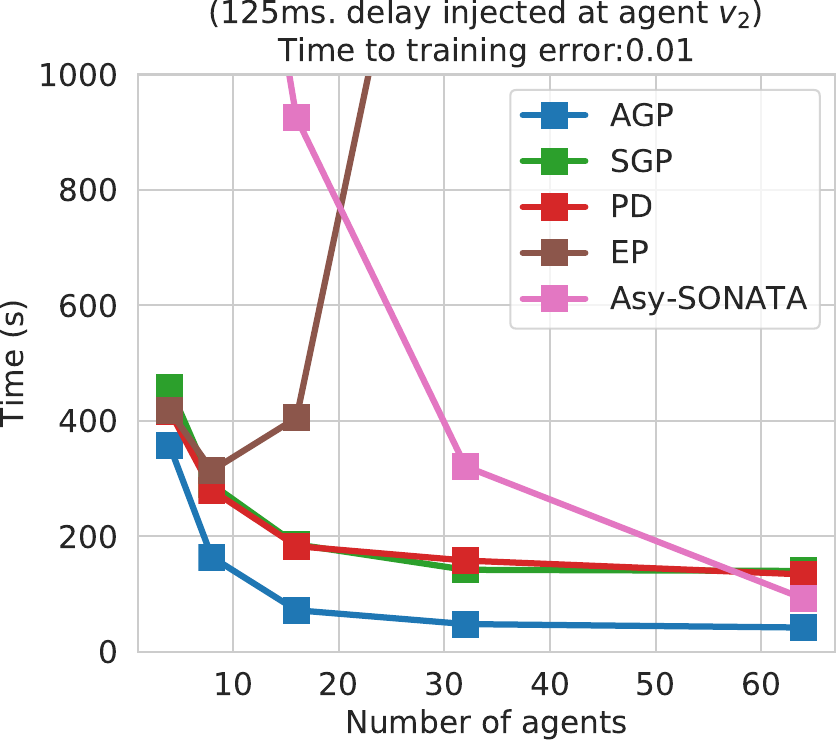}} \hspace*{0.005\textwidth}
\subfloat[\label{fig:250ms-world-sweep} $250$ms delay injected at agent $v_{2}$]{\includegraphics[width=0.24\textwidth]{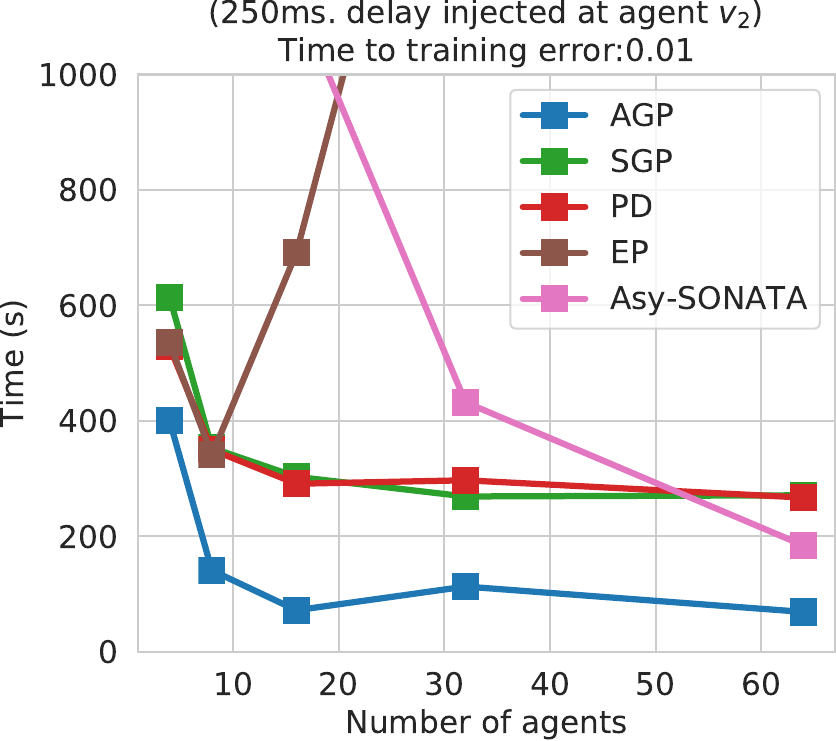}} \hspace*{0.005\textwidth}
\subfloat[\label{fig:500ms-world-sweep} $500$ms delay injected at agent $v_{2}$]{\includegraphics[width=0.24\textwidth]{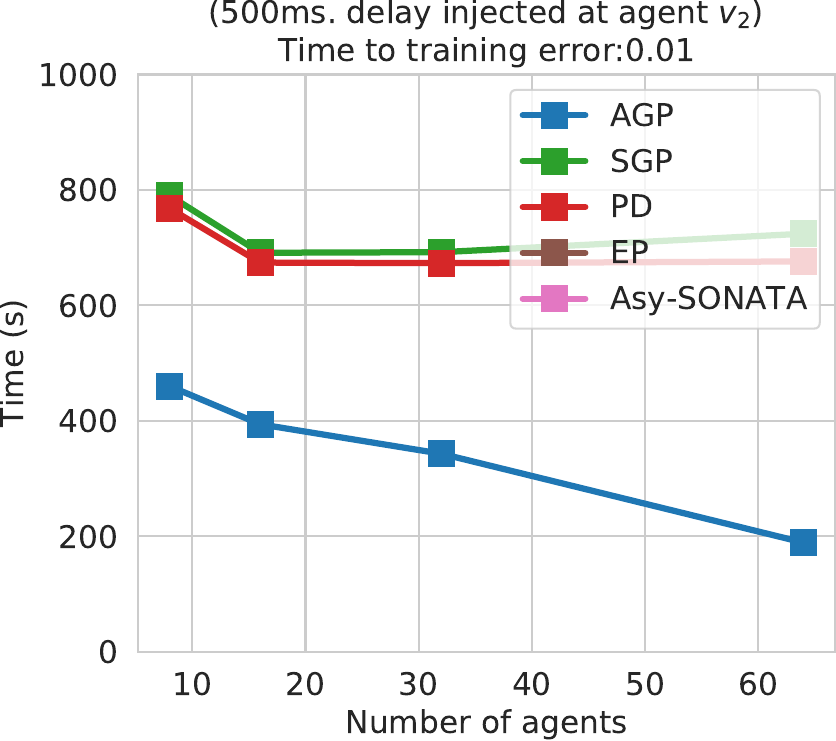}} 
\caption{Time $t[k]$ (seconds) at which $F( \xbar[k] ) - F(\xstar) < 0.01$ is satisfied for the first time in the Covertype experiments.
(a) Experiment run under normal operating conditions.
(b) An artificial $125$ms delay is injected at agent $v_2$ after every local iteration.
(b) An artificial $250$ms delay is injected at agent $v_2$ after every local iteration.
(c) An artificial $500$ms delay is injected at agent $v_2$ after every local iteration; neither EP nor Asy-SONATA obtained a residual error of $10^{-2}$ or below after $1000$s for this delay with any network size.
AGP reaches the threshold residual error $10^{-2}$ faster than all other methods.}
\label{fig:time_to_error}
\end{figure*}

Increasing $\tauratemax$ leads to a reduction in the iteration-wise convergence rate, as expected. However, increasing $\tauratemax$ also reduces the idling time, and thereby leads to an improvement in the time-wise convergence rate.
The dashed blue line in Fig.~\ref{fig:agp_train_tau} corresponds to the upper bound on $\norm{ \xstar_K - \xstar}$ from Theorem~\ref{th:bound_dist_minimziers}, where the values $\overline{p}^{(K)}_i$ are computed from the experiment corresponding to $\tauratemax=32$.
The dashed orange line corresponds to the true value of $\norm{ \xstar_K - \xstar}$, where the values $\overline{p}^{(K)}_i$ are also computed from the experiment corresponding to $\tauratemax=32$.

In Fig.~\ref{fig:x_K_sweep_tau} we plot the distance between the minimizer of the re-weighted objective and the original (unbiased) objective for each of the different choices of $\tauratemax$ used in this experiment.
As predicted from Theorem~\ref{th:bound_dist_minimziers}, the distance between minimizers decreases as the disparity in agent update rates decreases.

\subsection{Non-Synthetic Dataset}

To facilitate comparisons with existing methods in the literature, a regularized multinomial logistic regression classifier is trained on the \textit{Covertype} dataset~\cite{Lichman:2013} from the UCI~repository~\cite{uci}. Here the objective is to minimize, over model parameters $\bm{w}$, the negative log-likelihood loss function:
\begin{equation}
	F(\bm{w}) \defeq - \sum_{l=1}^D \sum^{K}_{j=1} \log \left( \frac{\exp(w_j^T x^l)}{\sum^K_{j^\prime=1} \exp(w_{j^\prime}^T  x^l) } \right)^{y^l_j} + \frac{\lambda}{2} \norm{\bm{w}}_{\bm{F}}^2, \label{eqn:multinomial-logistic}
\end{equation}
where $D = \num[group-separator={,}]{581012}$ is the number of training instances in the dataset, $K = 7$ is the number of classes, $x^l \in \R^{54}$ and $y^l \in \R^{7}$ correspond to the $l^{th}$ training instance feature and label vectors respectively (the label vectors are represented using a $1$-hot encoding), $\bm{w} \in \R^{7 \times 54}$ are the model parameters, and $\lambda > 0$ is a regularization parameter.
We take $\lambda = 10^{-4}$ in the experiments.
The $54$ features consist of a mix of categorical (binary $1$ or $0$) features and real numbers. We whiten the non-categorical features by subtracting the mean and dividing by the standard deviation. 

\begin{figure}[!t]
\centering
\subfloat[\label{fig:n16_delay} 16-agent Erd\H{o}s-R\'{e}nyi graph]{\includegraphics[width=0.48\textwidth]{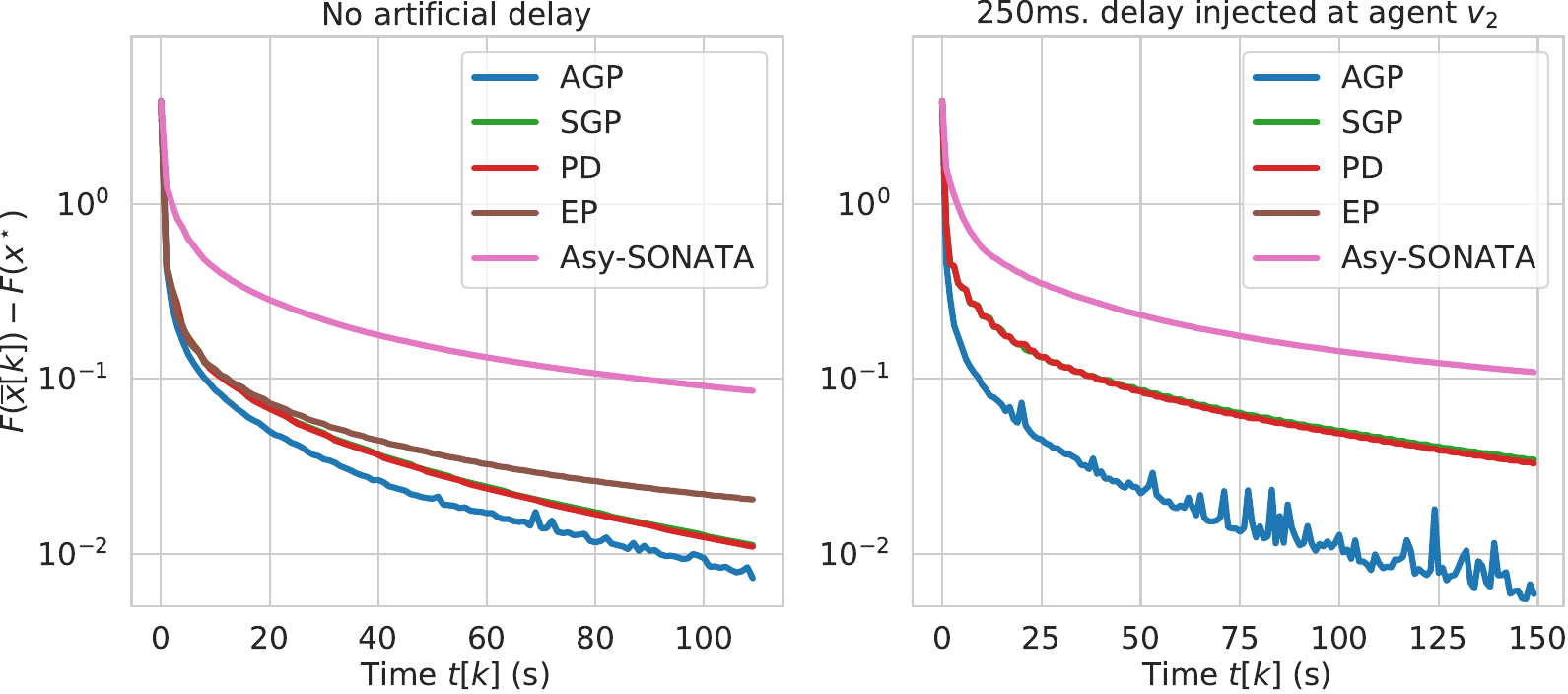}} \\
\subfloat[\label{fig:n64_delay} 64-agent Erd\H{o}s-R\'{e}nyi graph]{\includegraphics[width=0.48\textwidth]{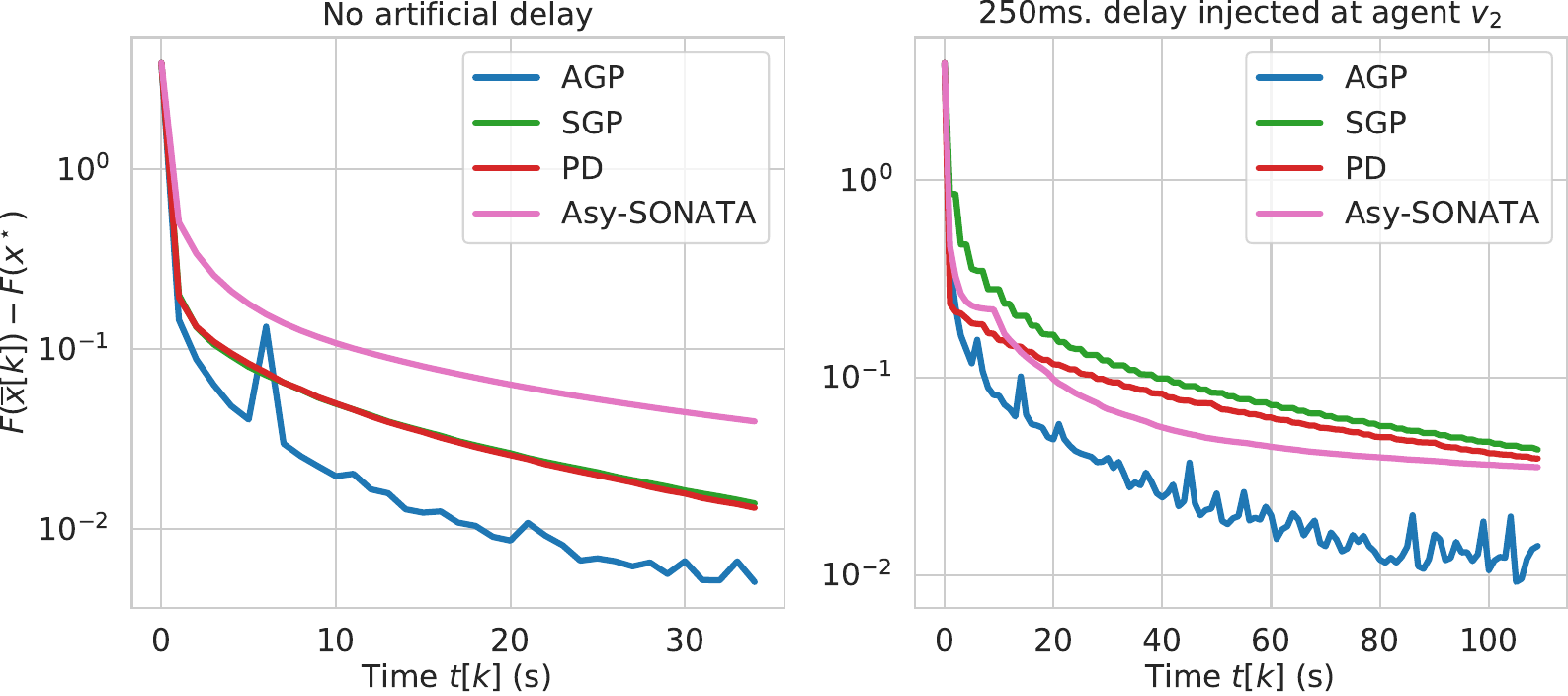}}
\caption{Multinomial logistic regression training error on the Covertype dataset using large multi-agent networks.
Left subplots in each figure correspond to normal operating conditions.
Right subplots correspond to experiments with an artificial $250$ms delay induced at agent $v_2$ at each local iteration.
(EP did not converge over the 64-agent network topology).
AGP is more robust than the synchronous algorithms to failing or stalling nodes.}
\label{fig:synch_asynch_delay_large}
\end{figure}

All network topologies are randomly generated using the Erd\H{o}s-R\'{e}nyi model where the expected out-degree of each agent is $4$, independent of $n$; \ie, with an edge probability of $\min\{4/(n - 1), 1\}$.
To investigate how the algorithms scale with the number of nodes, we consider different values of $n \in \{4, 8, 16, 32, 64\}$. In each case, we randomly partition the $D$ training instances evenly across the $n$ agents. 
All algorithms use a constant step-size, and we tuned the step-sizes separately for each algorithm using a simple grid-search over the range $\alpha \in [10^{-3}, 10^1]$. For all algorithms, the (constant) step-size $\alpha = 1.0$ gave the best performance.
Since the total number of samples $D$ is fixed, this problem has a fixed computational workload; as we increase the size of the network, the number of samples (and hence, the computational load) per agent decreases. 
The local objective function $f_i$ at agent $v_i$ is similar to that in \eqref{eqn:multinomial-logistic} but the sum over $l$ only involves those training instances assigned to agent $v_i$.

Fig.~\ref{fig:time_to_error} shows the first time $t[k]$ when the residual error satisfies $F(\xbar[k] ) - F(\xstar) < 0.01$, as a function of network size. 
Fig.~\ref{fig:0ms-world-sweep} shows that, under normal operating conditions, AGP decreases the residual error for both small and large network sizes faster than the state-of-the art methods and its synchronous counterpart.
To study robustness of the methods to delays, we run experiments where we inject an artificial delay at agent $v_2$ after every local iteration; the results are shown in Fig.~\ref{fig:125ms-world-sweep}, Fig.~\ref{fig:250ms-world-sweep}, and Fig.~\ref{fig:500ms-world-sweep} for $125$~ms, $250$~ms, and $500$ms delays, respectively.
To put the magnitude of these delays in context, Table~\ref{tb:update_times} reports the average agent update time for various network sizes.
As expected, we observe that asynchronous algorithms (AGP and Asy-SONATA) are more robust than the synchronous algorithms to slow nodes. However, for the $500$~ms delay case, Asy-SONATA did not achieve a residual error below $0.01$ after $1000$ seconds.
Fig.~\ref{fig:500ms-world-sweep} demonstrates that AGP is robust to such a large delay.

Fig.~\ref{fig:synch_asynch_delay_large} shows the residual error curves with respect to wall clock time for different network sizes, with and without an artificial $250$ms delay induced at agent $v_2$ at each iteration.
AGP is faster than the other methods under normal operating conditions (left subplots Fig.~\ref{fig:synch_asynch_delay_large}), and this performance improvement is especially pronounced when an artificial $250$ms delay is injected in the network (right subplots Fig.~\ref{fig:synch_asynch_delay_large}).
In the smaller multi-agent networks, a $250$ms delay is a relatively plausible occurrence.
In larger multi-agent networks a $250$ms delay is quite extreme since there could be over $\num[group-separator={,}]{2000}$ updates performed by the network in the time it takes the artificially delayed agent to compute a single update.
The fact that AGP is still able to converge in this scenario is a testament to its robustness.

\begin{table}[!t]
\caption{Average time taken by an agent to perform a gradient-based update for the Covertype experiments.}
\label{tb:update_times}
\vspace*{1em}
\centering
\begin{tabular}{ c l l l }
\hline
\textbf{\# agents} & Mean time (s) & Max.~time (s) & Min.~time (s) \\
\hline
4 & ${0.362}$ {\scriptsize $\pm 0.00649$} & 0.507 & 0.348 \\
8 & ${0.0993}$  {\scriptsize $\pm 0.0107$} & 0.139 & 0.0859 \\
16 & ${0.0488}$ {\scriptsize $\pm 0.00339$} & 0.0598 & 0.0430 \\
32 & ${0.0207}$ {\scriptsize $\pm 0.00166$} & 0.0284 & 0.0175 \\
64 & ${0.00849}$ {\scriptsize $\pm 0.000246$} & 0.0123 & 0.00797 \\ \hline
\end{tabular}
\end{table}

\section{Conclusion}
\label{sec:conclusion}

Our analysis of asynchronous Gradient-Push handles communication and computation delays. We believe our results could be extended to also deal with dropped messages using the approach described in~\cite{hadjicostis2016robust}, in which dropped messages appear as additional communication delays, which are easily addressed in our analysis framework.

Corollary~\ref{cor:congergence_opt_known_udpate_rates} showed that when agents know their relative update rates, then asynchronous Gradient-Push can be made to converge to the minimizer of $f$ rather than that of the reweighted objective~\eqref{eqn:F_K} by appropriately scaling the step-size. After the initial preprint of this work appeared online~\cite{assran2018preprint}, a related method was proposed in~\cite{Zhang2018asyspa} to estimate and track the update rates in a decentralized manner at the cost of additional communication overhead. Another related method was proposed in~\cite{tian2018achieving} that uses gradient tracking in combination with two sets of robust, asynchronous averaging updates --- one row stochastic, the other column stochastic --- to achieve provably geometric convergence rates at the cost of additional communication overhead and storage at each agent.

While extending synchronous Gradient-Push to an asynchronous implementation has produced considerable performance improvements, it remains the case that Gradient-Push is simply a multi-agent analog of gradient descent, and it would be interesting to explore the possibility of extending other algorithms to asynchronous operation using singly-stochastic consensus matrices; \eg, exploring methods that use an extrapolation between iterates to accelerate convergence; or quasi-Newton methods that approximate the Hessian using only first-order information; or Lagrangian-dual methods that formulate the consensus constrained optimization problems using the Lagrangian, or Augmented Lagrangian, and simultaneously solve for both primal and dual variables. Furthermore, it would be interesting to establish convergence rates for asynchronous versions of these algorithms.

Lastly, we find that, in practice, agents can asynchronously and independently control the upper bound on their relative processing delays, $\tauratemax$, by using non-blocking barrier primitives, such as those available as part of the MPI-3 standard.
It may be interesting to explore treating this as an algorithm parameter, rather than something dictated by the environment, and decreasing the delay bound according to some local iteration schedule so that one can realize the speed advantages of asynchronous methods at the start of training, and obtain the benefits of synchronous methods as one approaches the minimizer.
For example, from Definition~\ref{def:reweighted_obj}, it is clear that $\norm{\xstar_{K}- \xstar} \rightarrow 0$ when $\tauratemax \to 0$.
We believe that this is another interesting direction of future work.

\bibliographystyle{abbrv}
\bibliography{../../review.bib}

\end{document}